\documentclass[twoside]{article}

\usepackage{pgfplots}
\usepackage{amsfonts}
\usepackage{stmaryrd}
\usepackage{amssymb}
\usepackage{euscript}
\usepackage{amsthm}
\usepackage{amsmath}

\usepackage{latexsym}
\usepackage{mathrsfs}
\usepackage{graphicx}
\usepackage{color}
\usepackage{dsfont}
\usepackage{bm}
\usepackage{booktabs}
\usepackage{multirow}
\usepackage{array}

\usepackage{enumitem}%

\usepackage[hidelinks]{hyperref}

\numberwithin{equation}{section}

\newtheorem{theorem}{Theorem}[section]
\newtheorem{lemma}{Lemma}[section]

\newtheorem{corollary}{Corollary}[section]
\newtheorem{remark}{Remark}[section]

\newtheorem{hypothesis}{Assumption}[section]

\usepackage{comment}

\newcommand{\wt}{\widetilde}

\newcommand{\I}{\mathds{1}}

\newcommand{\esssup}{\operatornamewithlimits{ess\,sup}}

\newcommand{\cH}{\mathcal{H}}
\newcommand{\cS}{\mathcal{S}}
\newcommand{\cK}{\mathcal{K}}

\newcommand{\cA}{{\mathcal A}}

\newcommand{\cF}{{\mathcal F}}

\newcommand{\cM}{{\mathcal M}}
\newcommand{\cT}{{\mathcal T}}

\newcommand{\FF}{{\mathbb F}}

\newcommand{\PP}{{\mathbb P}}

\newcommand{\EE}{{\mathbb E}}

\begin{document}
\title{Entropy-regularized penalization schemes and reflected BSDEs with singular generators \vskip35pt}

\author{Daniel Chee$\,^{a}$, Noufel Frikha$\,^{b}$ and Libo Li$\,^{a}$ \\ \\ \\ \\
\\ $^{a\,}$School of Mathematics and Statistics, University of New South Wales \\ Sydney, NSW 2052, Australia \\ \\
$^{b\,}$Universit\'e Paris 1 Panth\' eon-Sorbonne, Centre d’Economie de la Sorbonne, \\106 Boulevard de l’H\^opital, 75642 Paris Cedex 13, France\\ }

\maketitle
\vskip20pt
\begin{abstract}
This paper extends Chee \emph{et al.}~\cite{CFL2025} to continuous-time optimal stopping, focusing on American options in an exploratory setting. Our first contribution is an entropy-regularized penalization scheme, inspired by classical penalization techniques for reflected BSDEs. It yields a smooth approximation of the stopping rule, promotes exploration, and enables gradient-based learning methods. We prove well-posedness, convergence, and illustrate numerical performance in low-dimensional examples. Our second contribution analyzes the behaviour of the scheme as the penalization parameter grows, showing that the limit solves a reflected BSDE with a logarithmically singular generator, for which we establish existence and uniqueness via a monotone limit argument.
\end{abstract}

\newpage
\tableofcontents
\newpage

\section{Introduction}
The numerical resolution of optimal stopping problems has long been central in mathematical finance, with applications ranging from American option pricing to optimal liquidation and real options. Recent advances in machine learning and reinforcement learning (RL) have renewed interest in Monte Carlo-based methods, particularly in high‑dimensional and model‑agnostic settings. A growing literature exploits randomized stopping representations and entropy regularization to recast optimal stopping as a stochastic control problem amenable to modern machine learning or RL techniques; see, for example, \cite{BCJ2019, DD2024, DSXZ2024, DFX2024, D2023, RST2022, STG2023}.

In this work, we consider a continuous‑time optimal stopping problem on a filtered probability space $(\Omega,\mathcal{F},\mathbb{F},\mathbb{P})$. For a given payoff process $P$, the value process is
\begin{equation}
	V_t = \esssup_{\tau \in \mathcal{T}_{t,T}} \mathbb{E}[P_\tau \mid \mathcal{F}_t],
	\label{Vstop}
\end{equation}
where $\mathcal{T}_{t,T}$ denotes the set of $\mathbb{F}$‑stopping times valued in $[t,T]$. It is well known that $V$ is the Snell envelope of $P$ and admits a Doob-Meyer decomposition, which can be characterized through a reflected backward stochastic differential equation (RBSDE).

Our starting point is the randomized stopping representation of Gy\"ongy and \v{S}i\v{s}ka~\cite{GS2008}, which rewrites the optimal stopping problem as a control problem over non‑negative stopping intensities. Specifically,
\begin{align}
	V_t
	= \esssup_{\gamma \in \Lambda}
	\mathbb{E}\Big[
	P_T e^{-\int_t^T \gamma_u du}
	+ \int_t^T P_s \gamma_s e^{-\int_t^s \gamma_u du} \, ds
	\,\Big|\, \mathcal{F}_t
	\Big],
	\label{rstopping}
\end{align}
where $\Lambda$ denotes the set of admissible intensity controls. This representation underpins several recent RL-based approaches to optimal stopping and admits a natural formulation in terms of BSDEs. It also brings into sharp focus the main difficulty of the continuous-time setting: the optimal stopping intensity is typically degenerate, taking only the extreme values $0$ and $+\infty$. Such degeneracy entails a severe lack of regularity, which in turn gives rise to substantial theoretical and numerical challenges, particularly for gradient-based learning algorithms.

To address this issue, several recent works have introduced entropy-regularization within exploratory HJB equations and relaxed-control frameworks; see \cite{DD2024, DSXZ2024, DFX2024, D2023}. Our approach differs. Rather than relying on PDEs or a specific diffusion model, we work directly with BSDEs to introduce entropy-regularization (via a temperature parameter) through the standard penalization scheme, see El Karoui \emph{et al.}~\cite{EKPPQ1997}. This yields an entropy‑regularized numerical scheme that is model‑independent, data-driven and flexible enough to incorporate nonlinear market dynamics.

The paper has two objectives. First, we develop and analyze this scheme for continuous‑time optimal stopping. For fixed penalization and temperature parameters, we establish well‑posedness and study convergence as these parameters are sent to their limits. Under suitable scaling, the entropy‑regularized value process converges to the classical value $V$, with an explicit rate under additional regularity. We also design a Policy Improvement Algorithm (PIA) tailored to the regularized formulation.

Second, we examine the scheme when the penalization parameter tends to infinity while temperature remains fixed. In this limit, the scheme converges monotonically to a process $V^\lambda$, which we show solves a RBSDE with a logarithmically singular driver, a class that appears to be new. Related work on BSDEs and RBSDEs with singular drivers includes \cite{ B2019, BEO2017, BT2021, DE1992,  LRZ2024, WJ2025, Z2024, ZZF2021,ZZM2024}, though these approaches rely on domination or Itô-Krylov techniques. Our analysis instead uses a monotone limit argument in the spirit of Peng~\cite{P1999}. 


The paper is organized as follows. Section~\ref{erps} studies the entropy‑regularized penalization scheme, its convergence properties, and the associated PIA. Section~\ref{ERRBSDE} is devoted to the asymptotic analysis of the scheme, including the well-posedness of the limiting singular RBSDE and its probabilistic interpretation. Section~\ref{NEP} provides numerical experiments, and technical lemmas are collected in the appendix.

\section{Notations and Setup}
We work on a filtered probability space $(\Omega, \cF, \FF=(\cF_t)_{t\ge0}, \PP)$. Additional assumptions on $\FF$ will be stated when needed (see Assumption~\ref{hyp: Fcont}). We denote by $\mathcal{O}(\FF)$ (and $\mathcal{P}(\FF)$) the space of $\FF$‑optional (and predictable) processes, by $\cM$ the $\FF$‑martingales, and by $\cA^+$ the predictable, positive, increasing processes. We use $x\vee y=\max(x,y)$, $x^{+}=\max(x,0)$, and $x^-=\max(-x,0)$. Throughout, $C$ and $K$ denote generic positive constants.

The Banach space of square-integrable optional processes, denoted $\mathcal{S}^2$, the
space of square-integrable martingales, denoted $\mathcal{H}^2$, and the space of
square-integrable, predictable, increasing processes, denoted $\mathcal{K}^2$, are defined by
\begin{align*}
	\cS^2 & := \Big\{X \in \mathcal{O}(\mathbb{F}) : \mathbb{E}\big[\sup_{0\leq t \leq T} X_t^2\big] < \infty \Big\},\\
	\cH^2 & := \left\{M \in \mathcal{M}: \mathbb{E}\left[[M]_T\right] < \infty \right\},\\
	\cK^2  & := \big\{A \in \mathcal{A}^{+} : \mathbb{E}\big[A_T^2\big] < \infty \big\},
\end{align*}
where $[M]$ stands for the quadratic variation of $M$.
If the payoff process $P$ is c\`adl\`ag and satisfies suitable integrability conditions, the value process $V$ in \eqref{Vstop} belongs to class (D).
In this case, $V$ satisfies the RBSDE
\begin{align}
	V_t & = P_T - (M_T - M_t) + (A_T - A_t), \qquad t\in [0,T] \label{VRBSDE}\\
	V_t & \geq P_t \quad \mathrm{and} \quad \int^T_0 (V_{s-}- P_{s-}) dA_s = 0, \nonumber 
\end{align}
\noindent where $M$ is a uniformly integrable martingale and $A$ is a predictable increasing process. In particular, when $P \in \cS^2$, one has $(V,M,A)\in \cS^2\times\cH^2\times \cK^2$; see, for example, Steps~1-4 in the proof of Lemma~3.3 in Grigorova \emph{et al.}~\cite{GIOOQ2017}.

\section{Entropy-Regularized Penalization Scheme}\label{erps}

In this section, we introduce our entropy-regularized penalization approach for American options. The methodology builds upon the relaxed control framework developed for Bermudan options in our earlier work \cite{CFL2025}. A key distinction between Bermudan and American options is that, in the latter case, the optimal control $\gamma$ in \eqref{rstopping} may take the value $+\infty$. 


To address this issue, we combine entropy-regularization with the classical penalization technique commonly used in the theory of reflected BSDEs. This hybrid approach allows us to retain the tractability of entropy-based methods while controlling the singular behavior of the optimal control. For an alternative approach see our related work \cite{CFL2026}. Throughout this paper, we impose the following assumption. 

\begin{hypothesis}\label{A}
	The payoff process $P$ is positive, càdlàg and belongs to $\mathcal{S}^2$.
\end{hypothesis}

We begin by truncating the control space and restricting attention to bounded controls. For a fixed $n \geq 1$, let $\Lambda_n$ denote the set of $\mathbb{F}$-adapted control processes $\gamma$ taking values in $[0,n]$. We start by considering the BSDE
\begin{align}
	V^n_t & = P_T - (M^n_T- M^n_t) + \esssup_{\gamma \in \Lambda_n}  \int^T_t (P_{s} - V^n_{s})\gamma_s \,  ds. \label{penalized}
\end{align}

Formally, if a solution $(V^n,M^n)$ to \eqref{penalized} exists, the optimal control is given pointwise by
$
\gamma^*_s = n\,\mathbf{1}_{\{P_s - V^n_s > 0\}},
$
so that the value process satisfies
\begin{align}
	V^n_t & = P_T - (M^n_T- M^n_t) + \int^T_t   n(P_{s} - V^n_{s})^+ \, ds. 
	\label{pscheme}
\end{align}
Following \cite{CFL2025}, we introduce an entropy-regularized version of \eqref{penalized} by relaxing the control. For $\lambda \geq 0$, referred to as the \emph{temperature parameter}, we consider the BSDE
\begin{align}
	V^{\lambda,n}_t & = P_T \!- \!\int^T_t\! dM^{\lambda,n}_s + \esssup_{\pi\in \Pi_{n}}\left[ \int^T_t \!\int^n_0 (P_{s} - V^{\lambda,n}_{s})u \pi_s(u) - \lambda \pi_s(u) \ln \pi_s(u) du ds\right], \label{AmericanBSDE}
\end{align}
where $\Pi_n$ denotes the set of $\mathbb{F}$-adapted probability densities on $[0,n]$, namely
\begin{align*}
	\left\{ \pi = (\pi_{s}(u))_{s\in [0,T]} \! :\! (u,s) \mapsto \pi_{s}(u) \in \mathcal{B}([0,n])\otimes \cF_{s}, \pi(u) \geq 0 \text{ and} \int^n_0 \pi_{s}(u) du = 1\right\}.
\end{align*}
Assuming that \eqref{AmericanBSDE} is well-posed, the pointwise optimal control admits a Gibbs-type representation:
\begin{align}\label{optimalpi}
	\pi^*_s(u) = \frac{\frac{1}{\lambda}(P_s-V^{\lambda,n}_{s})}{e^{\frac{n}{\lambda}(P_s-V^{\lambda,n}_{s})}-1} e^{\frac{1}{\lambda}(P_s-V^{\lambda,n}_{s})u}, \quad u \in [0,n].
\end{align}

Our objective is now to study rigorously the entropy-regularized penalization scheme \eqref{AmericanBSDE}, with the optimal control \eqref{optimalpi} substituted in, and to analyze its relationship with the value process $V$ of the original optimal stopping problem.

Substituting \eqref{optimalpi} into the generator of \eqref{AmericanBSDE} yields
\begin{align*}
	&  \int^n_0 \pi^*_s(u) \left[(P_{s} - V^{\lambda,n}_{s})  u  -  (P_s-V^{\lambda,n}_{s})u - \lambda\ln \left(\frac{\frac{1}{\lambda}(P_s-V^{\lambda,n}_{s})}{e^{\frac{n}{\lambda}(P_s-V^{\lambda,n}_{s})}-1}\right)\right] \, du  \\
	& \qquad = -\lambda \ln \left(\frac{\frac{1}{\lambda}(P_s-V^{\lambda,n}_{s})}{e^{\frac{n}{\lambda}(P_s-V^{\lambda,n}_{s})}-1}\right) =  \lambda \ln\!\left(
	\frac{e^{\frac{n}{\lambda}(P_s - V^{\lambda,n}_{s})} - 1}
	{\frac{n}{\lambda}(P_s - V^{\lambda,n}_{s})}
	\right)
	+ \lambda \ln(n).
\end{align*}
To simplify notation, we introduce the functions
\begin{equation}\label{def:psi:phi}
	\Psi(x)  = \frac{1}{x}\ln \left( \frac{e^{x}-1}{x} \right)  \quad \mathrm{and} \quad \Phi(x)  = \ln \left( \frac{e^{x}-1}{x} \right), \quad x \in \mathbb{R}.
\end{equation}
The entropy-regularized BSDE \eqref{AmericanBSDE} can then be rewritten as
\begin{align}
	V^{\lambda,n}_t 
	& = P_T - \int^T_t dM^{\lambda,n}_s + \int^T_t \left[(P_s-V^{\lambda,n}_{s})n\Psi\left(\frac{P_s-V^{\lambda,n}_{s}}{\lambda/n}\right)ds  + \lambda \ln(n)\right]ds \label{Vnn}
\end{align}
or, equivalently,
\begin{align}		
	V^{\lambda,n}_t 	& = P_T - \int^T_t dM^{\lambda,n}_s + \int^T_t n\left[\frac{\lambda}{n} \Phi\left(\frac{P_s-V^{\lambda,n}_{s}}{\lambda/n}\right) + \frac{\lambda}{n}\ln(n)\right]ds. \label{Vln}
\end{align}
\begin{lemma} \label{lemma: Vln_well_posed}
	For each $n\ge 1$, the entropy-regularized penalization scheme \eqref{Vln} is well-posed, that is there exists a unique solution $(V^{\lambda,n}, M^{\lambda,n}) \in \cS^2 \times \cH^2$ to \eqref{Vln}.
\end{lemma}
\begin{proof}
	By Lemma \ref{lipphi} and Lemma \ref{lemma1.1}, the generator is Lipschitz continuous and, for any fixed $n$ and $\lambda$,
	$$\mathbb{E}\Big[\int^T_0 \left|\frac{\lambda}{n}\Phi\left(\frac{P_s}{\lambda/n}\right)\right|^2ds \Big]  <\infty.$$
	Existence and uniqueness then follow from Theorem 3.1 in \O ksendal and Zhang \cite{OZ}.	
\end{proof}

\subsection{Convergence to the American Option}
In this subsection, we study the asymptotic behaviour of the entropy‑regularized penalization scheme \eqref{Vnn}. For any fixed $n \ge 1$, we show that $V^{\lambda,n}$ converges to the classical penalized value $V^n$ in \eqref{pscheme} as $\lambda \downarrow 0$. Moreover, Theorem~\ref{t3.1} implies that if the truncation parameter satisfies $\lambda \ln(n) \to 0$, then $V^{\lambda,n}$ converges to the American option value $V$ as $\lambda \to 0$. This result forms the basis of our numerical approach.

Define $d\Gamma^{\lambda,n}_s :=  n\Psi\big(n(P_s - V^{\lambda,n}_s)\lambda^{-1}\big) ds$. By Lemma~\ref{lem:psi:cdf}, $\Gamma^{\lambda,n}$ is a bounded, non-decreasing hazard process as the integrand is non-negative and bounded by~$n$. Applying It\^o's formula to $e^{-\Gamma^{\lambda,n}_t}(V^{\lambda,n}_t + \lambda\ln(n)t)$ and setting $\widehat P_t := P_t + \lambda \ln(n)t$ yields 
\begin{align*} 
	V^{\lambda,n}_t +\lambda \ln(n) t \nonumber 
	& = \mathbb{E}\Big[\widehat P_Te^{-(\Gamma^{\lambda,n}_T - \Gamma^{\lambda,n}_t)} + \int_{]t,T]} \widehat P_u e^{-(\Gamma^{\lambda,n}_u - \Gamma^{\lambda,n}_t)} d\Gamma_u^{\lambda,n}\,\Big|\cF_t\Big].
\end{align*}
Using the link between randomized and classical optimal stopping \eqref{rstopping}, we deduce
\begin{gather}
	V^{\lambda,n}_t +\lambda \ln(n) t \leq \esssup_{\tau \in \mathcal{T}_{t,T}}\mathbb{E}\big[P_{\tau\wedge T} + \lambda \ln(n)(T\wedge \tau)\,\big|\cF_t\big] \leq V_t +\lambda \ln(n) T. \label{AmericanUpper}
\end{gather}

Hence, if the penalization parameter $n$ is chosen in such a way that $\lambda \ln (n) \rightarrow 0$ as $\lambda \rightarrow 0$, the approximation $V^{\lambda,n}$ converges to the American value process $V$ as $\lambda\rightarrow 0$. The presence of the correction term $\lambda\ln(n)(T-t)$ in \eqref{Vnn} prevents a direct comparison with the classical penalization scheme. To overcome this difficulty, we introduce the auxiliary BSDE
\begin{align*}
	\wt V^{\lambda,n}_t &  =  P_T - \int^T_t dM^{\lambda,n}_s + \int^T_t n\Big[\frac{\lambda}{n} \Phi\Big(\frac{P_s-\wt V^{\lambda,n}_{s}}{\lambda/n}\Big)\Big]ds. 
\end{align*}
This equation is well-posed, since the driver is Lipschitz continuous with constant $n$ by Lemma~\ref{lipphi}, and $P \in \mathcal{S}^2$. Setting $d\wt \Gamma_s^{\lambda,n} = n\Psi\big(n(P_s-\wt V^{\lambda,n}_{s})\lambda^{-1}\big)ds$ and applying It\^o’s formula to $e^{-\wt \Gamma^{\lambda,n}_t}\widetilde V^{\lambda,n}_t$, we obtain
\begin{align*}
	\wt V^{\lambda,n}_t & = \mathbb{E}\Big[P_Te^{-(\wt \Gamma^{\lambda,n}_T - \wt\Gamma^{\lambda,n}_t) } + \int^T_t P_u e^{-(\wt \Gamma^{\lambda,n}_u - \wt\Gamma^{\lambda,n}_t)}  d \wt \Gamma_u^{\lambda,n} \Big | \cF_{t}\Big], \quad t \in [0,T].
\end{align*}
By the comparison theorem for BSDEs (see Theorem~3.4 in \cite{OZ}), for all $t\in [0,T]$
\begin{equation}\label{first:simple:comparison:bsde:I}
	\wt V^{\lambda,n}_t \leq V^{n}_t\leq V_t, \quad a.s.
\end{equation}
\noindent recalling that $V^n$ denotes the classical penalization scheme introduced in \eqref{pscheme}.

\begin{remark} 
	The process $\widetilde V^{\lambda,n}$ admits a natural variational interpretation. 
	Instead of using an entropy term of the form $\lambda\,\pi_s(u)\ln \pi_s(u)$, the entropy term is given by 
	the Kullback--Leibler divergence of $\pi_s$ with respect to the uniform distribution on $[0,n]$, 
	namely $\lambda\,\pi_s(u)\ln\!\big(\frac{\pi_s(u)}{1/n}\big)$. 
\end{remark} 

\begin{theorem}\label{t3.1}
	For any positive integer $n$, any $\lambda \in (0,1]$ and $t\in [0,T]$, we have
	\begin{align*}
		|V^n_t - V^{\lambda,n}_t| 
		& \leq C_{t,T}\big(\lambda - \lambda\ln(\lambda) + \lambda \ln(n)\big), \quad a.s.
	\end{align*}
	where $C_{t,T} = 2(T-t)\big(1- \ln(1-e^{-1}) + \mathbb{E}\big[\sup_{s\in [t,T]}\ln( V_s\vee 1) \big|\cF_t\big]\big)$.
\end{theorem}
\begin{proof}
	
	\noindent \emph{Step 1.} By taking the difference between $V^n$ and $\wt V^{\lambda,n}$, we obtain
	\begin{equation*}
		V^{n}_t - \wt V^{\lambda, n}_t = \int_{]t,T]}d(M^{\lambda, n} - M^{n})_{s} + \int_{[t,T[} n \left[(P_{s} - V_{s}^{n})^{+} - \frac{\lambda}{n}\Phi\left(\frac{P_{s} - \wt V_{s}^{\lambda, n}}{\lambda / n}\right)\right] ds.
	\end{equation*}
	Taking $c = \epsilon = \lambda/n$ in Lemma \ref{lemma1.1} and using \eqref{first:simple:comparison:bsde:I}, we get
	\begin{align*}
		V^{n}_t  &- \wt V^{\lambda, n}_t  
		\leq \mathbb{E}\left[ \int_{[t,T[} n \left[(P_{s} - V_{s}^{n})^{+} - \frac{\lambda}{n}\Phi\left(\frac{P_{s} - V_{s}^{n}}{\lambda / n}\right)\right] ds \right] \\
		& \leq \mathbb{E}\left[ \int_{[t,T[}  \left[\lambda - \lambda\ln(1-e^{-1})+ \lambda\ln(|P_{s} - V^n_{s}|\vee 1) - \lambda\ln(\lambda) + \lambda\ln(n) \right] \Bigg|\cF_t \right]\\
		& \leq (T-t)\Big(1- \ln(1-e^{-1}) + \mathbb{E}\Big[\sup_{s \in [t,T]} \ln( V_s\vee 1) \Big|\,\cF_t\Big]\Big) \big(\lambda - \lambda\ln(\lambda) + \lambda \ln(n) \big). 
	\end{align*}
	
	\emph{Step 2.} On the other hand, from the comparison theorem for BSDEs, see again Theorem 3.4 in \cite{OZ}, we know that $\wt V^{\lambda,n} \leq V^{\lambda,n}$. Hence, since $x \mapsto \Phi(c^{-1}(P-x))$ is non-increasing, we obtain
	\begin{align*}
		V^{\lambda,n}_t - \wt V^{\lambda, n}_t  
		& = \mathbb{E}\left[\int^T_t \lambda\left[ \Phi\left(\frac{P_s- V^{\lambda,n}_{s}}{\lambda/n}\right) - \Phi\left(\frac{P_s- \wt V^{\lambda,n}_{s}}{\lambda/n}\right) + \ln(n) \right]ds  \,\Bigg|\,\cF_t \right]\\
		& \leq \lambda\ln(n) (T-t).
	\end{align*}
	Combining the previous bounds yields
	\begin{align*}
		& |V^n_t - V^{\lambda,n}_t|  \leq |V^n_t - \wt V^{\lambda,n}_t| + |\wt V^{\lambda,n}_t - V^{\lambda,n}_t|\\
		& \leq (T-t)\Big(1- \ln(1-e^{-1}) + \mathbb{E}\Big[\sup_{s \in [t,T]} \ln( V_s\vee 1) \Big|\cF_t\Big]\Big) \big(\lambda - \lambda\ln(\lambda) + 2\lambda \ln(n)\big).
	\end{align*} 
\end{proof}

Finally, classical arguments (see Steps~2--3 of Theorem~4.1 in \cite{HO2015}) ensure that $V^n\to V$ as $n\to\infty$. Theorem~\ref{t3.1} implies that
if $P$ is bounded and $n=1/\lambda$, then $$|V^\frac{1}{\lambda}_t - V^{\lambda,\frac{1}{\lambda}}_t| \leq C(\lambda - \lambda \ln \lambda).$$

To obtain an explicit convergence rate towards $V$, estimates of the classical penalization error $V-V^n$ are required. Such results typically require additional structure on the payoff process $P$ and the filtration $\mathbb{F}$. In the Brownian setting, one may impose the following assumption (see El Karoui \emph{et al.} \ \cite{EKPPQ1997} and Gobet and Wang\ \cite{GWX2023}): 

\begin{hypothesis}\label{hypP2}
	The filtration $\FF$ is a Brownian filtration, and the payoff process $P$ admits a generalized semimartingale decomposition as  
	\begin{equation*}
		P_{t} = P_{0} + \int_{0}^{t}U_{s}ds + \int_{0}^{t}V_{s}dW_{s} + H_{t},
	\end{equation*}
	
	\noindent where $U, V \in \cS^{2}$, and $H$ is continuous, non-decreasing, with $H_{T} \in L^{2}$ and $H_{0} = 0$.
\end{hypothesis}

\begin{corollary}\label{scheme1rate}
	Assume that $P$ is bounded and that Assumption \ref{hypP2} holds. If $$\overline{\kappa}_{\infty} := \esssup_{(t,\omega) \in [0,T]\times \Omega} U_{t}^{-}(\omega) <\infty,$$ then we have
	\begin{gather*}
		\sup_{0\leq t \leq T}|V_t - V^{\lambda, \frac{1}{\lambda}}_t| \leq C\overline{\kappa}_\infty (\lambda - \lambda \ln \lambda), \quad a.s.
	\end{gather*}
\end{corollary}
\begin{proof}
	By Theorem 3.5 (and Remarks 3.6-3.7) in \cite{GWX2023}, 
	$
	0\leq V_t - V^n_t \leq \overline{\kappa}_\infty n^{-1}.
	$
	Choosing $n=1/\lambda$ and combining with Theorem \ref{t3.1} yields the result.
\end{proof}

\subsection{Policy Improvement Algorithm}
In this subsection, we present a PIA for computing $V^{\lambda, n}$ and analyze its convergence. Policy improvement methods are classical tools in stochastic control, providing a constructive iterative procedure that alternates between policy optimization and value evaluation, often leading to fast and numerically stable convergence.

Fix $\lambda\in(0,1]$ and let $\pi=(\pi_s)_{s\in[0,T]}\in\Pi_n$ be a conditional density. Define
\begin{align*}
	G(s, x, \pi_{s}) & := \int^n_0 \Big\{(P_{s} - x)u \pi_{s}(u) - \lambda \pi_{s}(u)  \ln(\pi_{s}(u) )\Big\}du,
\end{align*}
and denote by
\begin{equation} \label{def: pia_pi_star}
	\pi^{*}_{s}(x,u) : = \textrm{argmax}_{\pi} G(s, x, \pi_{s}(u)) = \frac{\frac{1}{\lambda}(P_{s}-x)}{e^{\frac{n}{\lambda}(P_{s}-x)}-1} e^{\frac{1}{\lambda}(P_{s}-x)u},\quad  u \in [0,n],
\end{equation}
\noindent the optimal policy associated with the state $x$.

Motivated by Corollary~\ref{scheme1rate}, we fix $n:=1/\lambda$ and set
$
\mathscr{V}^\lambda : = V^{\lambda, \frac{1}{\lambda}}$. Let $\mathscr{V}^{\lambda,0}$ be an initial guess, e.g. 
$$
\mathscr{V}^{\lambda,0}_t := \mathbb{E}[P_T |\mathcal{F}_t], \quad t \in [0,T].
$$ 

Given $\mathscr{V}^{\lambda,m}$, the $(m+1)$-th iteration consists of:

\medskip
\noindent\emph{Policy update.}
\begin{equation}\label{eq: pin2}
	\pi^{m+1}_s(u)
	:= \pi_s^*(\mathscr{V}^{\lambda,m}_s,u)
	= \frac{\frac{1}{\lambda}(P_s-\mathscr{V}^{\lambda,m}_s)}
	{e^{\frac{n}{\lambda}(P_s-\mathscr{V}^{\lambda,m}_s)}-1}
	e^{\frac{1}{\lambda}(P_s-\mathscr{V}^{\lambda,m}_s)u},
	\qquad u\in[0,n].
\end{equation}

\noindent\emph{Policy evaluation.}
\begin{equation}\label{eq: pin3}
	\mathscr{V}^{\lambda,m+1}_t
	= P_T - (\mathscr{N}^{\lambda,m+1}_T-\mathscr{N}^{\lambda,m+1}_t)
	+ \int_t^T G(s,\mathscr{V}^{\lambda,m+1}_s,\pi^{m+1}_s)\,ds,
\end{equation}
where
\begin{equation}\label{eq: Gpi_def}
	G(s,\mathscr{V}^{\lambda,m+1}_s,\pi^{m+1}_s)
	= \lambda\Phi\!\left(\frac{P_s-\mathscr{V}^{\lambda,m}_s}{\lambda/n}\right)
	+ \lambda\ln n
	+ (\mathscr{V}^{\lambda,m}_s-\mathscr{V}^{\lambda,m+1}_s)\,
	\mu_{\pi^{m+1}_s},
\end{equation}
and
\begin{equation}\label{mu}
	\mu_{\pi^{m+1}_s}
	= \int_0^n u\,\pi^{m+1}_s(u)\,du
	= \frac{n}{1-e^{-\alpha_s^m n}}-\frac{1}{\alpha_s^m},
	\qquad
	\alpha_s^m := \frac{P_s-\mathscr{V}^{\lambda,m}_s}{\lambda}.
\end{equation}

Each iteration requires solving
the linear BSDE \eqref{eq: pin3}, and computing $\pi^{m+1}\in\Pi_n$ defined by \eqref{eq: pin2}. Moreover, since
\[
G(s,\mathscr{V}^{\lambda,m}_s,\pi^{m+1}_s)
= G(s,\mathscr{V}^{\lambda,m}_s,\pi_s^*(\mathscr{V}^{\lambda,m}_s))
\ge G(s,\mathscr{V}^{\lambda,m}_s,\pi^m_s),
\]
the comparison theorem for BSDEs (Theorem~3.4 in \cite{OZ}) implies that for any integer $m$ and any $t\in [0,T]$
\[
\mathscr{V}^{\lambda,m+1}_t\ge \mathscr{V}^{\lambda, m}_t,
\quad a.s.
\]
so that the sequence $(\mathscr{V}^{\lambda,m}_t)_{m\geq 1}$ is monotonically non-decreasing.

\begin{remark}
	If the term $\lambda\ln n$ is removed from the driver \eqref{eq: Gpi_def}, the same construction yields a PIA for $\widetilde V^{\lambda,1/\lambda}$, with $\widetilde V^{\lambda,1/\lambda}_t\le V_t$, $t\in [0,T]$.
\end{remark} 

\begin{theorem}  \label{thm: policy_convergence}
	For any fixed $\lambda \geq 0$ and $t \in [0,T)$, it holds
	\begin{gather*}
		0\leq \mathscr{V}_{t}^{\lambda} - \mathscr{V}_{t}^{\lambda,m} \leq \frac{(nT)^m}{m!} \mathbb{E}[\sup_{s\leq T} (\mathscr{V}^{\lambda, 1}_s - \mathscr{V}^{\lambda, 0}_s)\,|\, \cF_t], \quad a.s.
	\end{gather*}
	where $n = 1/\lambda$ is the truncation parameter.
\end{theorem}
\begin{proof}
	Fix $t \in [0,T)$ and consider
	\begin{align} \label{eq: policy_convergence_eq1}
		\mathscr{V}_{t}^{\lambda, m+1} - \mathscr{V}_{t}^{\lambda,m} 
		& = -(\mathscr{N}_{T}^{\lambda,m+1} - \mathscr{N}_{T}^{\lambda,m}) + (\mathscr{N}_{t}^{\lambda, m+1} - \mathscr{N}_{t}^{\lambda,m}) \nonumber\\
		& \quad + \int^T_t \left \{ G(s, \mathscr{V}_{s}^{\lambda, m+1}, \pi_{s}^{m+1}) - G(s, \mathscr{V}_{s}^{\lambda, m}, \pi_{s}^{m})  \right\}ds.
	\end{align}
	Since $\mathscr{V}^{\lambda,m+1} \geq \mathscr{V}^{\lambda,m}$ and $ \Phi$ is non-decreasing,
	\begin{align} \label{eq: policy_convergence_eq2}
		& G(s, \mathscr{V}_{s}^{\lambda, m+1}, \pi_{s}^{m+1}) - G(s, \mathscr{V}_{s}^{\lambda, m}, \pi_{s}^{m}) \nonumber \\
		& = \lambda \Phi\left(\frac{P_{s} - \mathscr{V}_{s}^{\lambda, m}}{\lambda/n}\right) - \lambda \Phi\left( \frac{P_{s} - \mathscr{V}_{s}^{\lambda, m-1}}{\lambda/n}\right) \\
		& \qquad + (\mathscr{V}_{s}^{\lambda, m} - \mathscr{V}_{s}^{\lambda, m+1}) \mu_{\pi_{s}^{m+1}} - (\mathscr{V}_{s}^{\lambda, m-1} - \mathscr{V}_{s}^{\lambda, m}) \mu_{\pi_{s}^{m}} \nonumber \\
		& \leq n(\mathscr{V}_{s}^{\lambda, m} - \mathscr{V}_{s}^{\lambda, m-1}).
	\end{align}
	
	By the Fubini-Tonelli theorem and iterating, we obtain
	\begin{align*}
		0\leq \mathscr{V}^{\lambda, m+1}_t - \mathscr{V}^{\lambda, m}_t & = \mathbb{E}[\int^T_t \left\{ G(s, \mathscr{V}_{s}^{\lambda, m+1}, \pi_{s}^{m+1}) -  G(s, \mathscr{V}_{s}^{\lambda, m}, \pi_{s}^{m})\right\} \, ds   \,|\, \cF_t]\\
		& \leq n\mathbb{E}[\int^T_t (\mathscr{V}^{\lambda, m}_{t_{m-1}} - \mathscr{V}^{\lambda, m-1}_{t_{m-1}})d{t_{m-1}}   \,|\, \cF_t]\\
		& \leq n^m\mathbb{E}[\int^T_{t}\int^T_{t_{m-1}}\dots \int^T_{t_1} (\mathscr{V}^{\lambda, 1}_{t_0} - \mathscr{V}^{\lambda, 0}_{t_0})\, dt_0\dots dt_{m-2}dt_{m-1}  \,|\, \cF_t]\\
		& \leq \frac{(nT)^m}{m!} \mathbb{E}\big[\sup_{s\leq T} (\mathscr{V}^{\lambda, 1}_{s} - \mathscr{V}^{\lambda, 0}_s)\,\big|\, \cF_t\big].
	\end{align*}
	
	To obtain the convergence to $ \mathscr{V}^{\lambda}$, we observe that 
	\begin{align*}
		G(s, \mathscr{V}_{s}^{\lambda}, \pi^*_{s}(\mathscr{V}^\lambda_s))&  - G(s, \mathscr{V}_{s}^{\lambda, m+1}, \pi_{s}^{m+1}) 
		\\
		& = \lambda \Phi\left(\frac{P_{s} - \mathscr{V}_{s}^{\lambda}}{\lambda/n}\right) - \lambda \Phi\left( \frac{P_{s} - \mathscr{V}_{s}^{\lambda, m}}{\lambda/n}\right)   - (\mathscr{V}_{s}^{\lambda, m} - \mathscr{V}_{s}^{\lambda, m+1}) \mu_{\pi_{s}^{m}} \nonumber \\
		& \leq n(\mathscr{V}_{s}^{\lambda, m+1} - \mathscr{V}_{s}^{\lambda, m}).
	\end{align*}
	Hence, from similar computations and the observation that $G(s, \mathscr{V}^{\lambda, m+1}_{s}, \pi_{s}^{m+1}) \leq G(s, \mathscr{V}_{s}^{\lambda, m+1}, \pi^*_s(\mathscr{V}_s^{\lambda, m+1}))$, we conclude that
	\begin{align*}
		0\leq \mathscr{V}^{\lambda}_t - \mathscr{V}^{\lambda, m+1}_t & = \mathbb{E}\Big[\int^T_t \left\{ G(s, \mathscr{V}_{s}^{\lambda}, \pi^*_{s}(\mathscr{V}^\lambda_s)) -  G(s, \mathscr{V}_{s}^{\lambda, m+1}, \pi_{s}^{m+1}) \right\} \, ds   \,\Big|\, \cF_t\Big]\\
		& \leq \frac{(nT)^{m+1}}{(m+1)!} \mathbb{E}\big[\sup_{s\leq T} (\mathscr{V}^{\lambda, 1}_{s} - \mathscr{V}^{\lambda, 0}_s)\,\big|\, \cF_t\big].
	\end{align*}
	Choosing $n = 1/\lambda$ concludes the proof.
\end{proof}

\section{Limit of the Entropy-Regularized Penalization Scheme}\label{ERRBSDE}

This section analyzes the asymptotic behavior of the entropy-regularized
penalization scheme \eqref{AmericanBSDE} as the truncation parameter
$n \to \infty$ for fixed $\lambda \in (0,1]$. Our aims are threefold:
(i) to study the convergence of the scheme and quantify the approximation
error induced by entropy regularization; (ii) to identify the limit as
the value component of a reflected BSDE with a logarithmically singular
generator; and (iii) to give a financial interpretation of this limiting
formulation, linking entropy regularization to endogenous default risk
and early-exercise behavior. Although not directly aimed at numerical
pricing, these results provide a connection between our singular
RBSDEs and risk-sensitive optimal stopping problems.

For $n \ge 1$ we introduce the continuous function
\begin{equation}\label{def:Phi:n}
	\Phi_n(x)
	:= \ln\!\left( \frac{e^{nx}-1}{x} \right),
	\qquad x \in \mathbb{R}\setminus\{0\},
\end{equation}
and, by continuity, we extend this definition at the origin by setting $\Phi_n(0) := \ln(n)$, and let $\Phi_n^{-1}$ denote its inverse (see Lemma \ref{lem:prop:phi:n}). Furthermore, we set
\begin{equation}\label{def: phi_lambda_n}
	\Phi_{\lambda, n}(s,x)
	:= \lambda \ln\!\left( \frac{e^{\,n(P_s - x)/\lambda}-1}{(P_s - x)/\lambda} \right)
	= \lambda\, \Phi_n\!\left( \frac{P_s - x}{\lambda} \right)
\end{equation}

\noindent and we observe that
\begin{equation}\label{limit:driver}
	\begin{aligned}
		\lim_{n\rightarrow\infty}\Phi_{\lambda,n}(s,x) = \Phi_{\lambda,\infty}(s,x)
		:= \begin{cases} 
			\lambda \ln\!\left( \frac{\lambda}{x - P_s} \right)\,
			& \quad  x > P_s, \\
			\infty  & \quad  x \le P_s.
		\end{cases}
	\end{aligned}
\end{equation}

The entropy-regularized penalization scheme $V^{\lambda,n}$ in \eqref{Vln} takes the form
\begin{equation}\label{vn}
	V^{\lambda,n}_t
	= P_T
	- \int_t^T dM_s^{\lambda,n}
	+ \int_t^T \Phi_{\lambda,n}(s, V^{\lambda,n}_s)\, ds ,
\end{equation}

\noindent which is also well-posed since $x \mapsto \Phi_{\lambda,n}(s,x)$ is Lipschitz continuous and $P \in \mathcal{S}^2$, see e.g. \cite{OZ}. Furthermore, for $x \in \mathbb{R}\setminus\{0\}$ and $n \ge 1$,
\begin{equation}\label{deriv:phin}
	\frac{d}{dn} [ \Phi_n(x) ]
	= \frac{d}{dn} \Big[ \ln\!\left( \frac{e^{nx}-1}{x} \right) \Big]
	= \frac{x e^{nx}}{e^{nx}-1}
	= \frac{x}{1 - e^{-nx}}
	\ge 0,
\end{equation}
and $\frac{d}{dn} [ \Phi_n(0) ]
= \frac{1}{n}
\ge 0$. Hence, for any $x \in \mathbb{R}$ and any $n \ge 1$, we have $\Phi_{\lambda, n}(s, x) \leq \Phi_{\lambda, n+1}(s, x)$. By the comparison theorem (Theorem~3.4 in~\cite{OZ}), this monotonicity implies that for all integer $n$ and all $t\in [0,T]$
\begin{equation}\label{increasing:sequence:vlambdan}
	V^{\lambda,n}_t \;\le\; V^{\lambda,n+1}_t,
	\qquad \;\text{a.s.}
\end{equation}
Therefore, for all $t\in [0,T]$, the limit
\begin{equation}\label{existence:limit:Vlambda}
	V^{\lambda}_t := \lim_{n\to\infty} V^{\lambda,n}_t
\end{equation}
exists almost surely.


We now construct the limit process $V^\lambda$. The argument follows a
monotone stability method for BSDEs, in the spirit of Peng~\cite{P1999},
adapted to the entropy-regularized driver. Using the monotonicity of the family
$(\Phi_{\lambda,n})_{n\ge1}$ and suitable uniform estimates, we pass to the
limit in the penalization scheme \eqref{Vln}, identifying the monotone limit
$V^\lambda$ as the value component of a reflected BSDE with a logarithmically
singular generator. Under Assumption~\ref{hyp: Fcont}, Theorem~\ref{t4.2}
shows that there exist $M^\lambda \in \mathcal{H}^2$ and
$A^\lambda \in \mathcal{K}^2$ such that
$(V^\lambda, M^\lambda, A^\lambda)$ is the unique solution in
$\mathcal{S}^2 \times \mathcal{H}^2 \times \mathcal{K}^2$ of
\begin{equation}\label{entropyrbsde}
	\begin{aligned}
		& V^\lambda_t = P_T - \int^T_t dM^\lambda_s + \int^T_t  \Phi_{\lambda, \infty}(s,V^\lambda_s) \, ds + A^\lambda_T-A^\lambda_t,\nonumber \\
		& V^\lambda_t \geq P_t \quad \mbox{ for all } t\in [0,T], \quad \mathrm{and} \quad \int^T_0 (V^\lambda_{s} - P_{s}) dA^\lambda_s = 0, 
	\end{aligned}
\end{equation}
We emphasize that the generator $\Phi_{\lambda,\infty}(s,\cdot)$ is \emph{not}
globally Lipschitz on $[P_s,\infty)$ due to its logarithmic singularity as
$x \downarrow P_s$, making \eqref{entropyrbsde} a non-standard reflected
BSDE. Existence is therefore not immediate. However, the generator is
monotone, or one-sided Lipschitz, and uniqueness follows from arguments
similar to those in Theorem 2.1 of Lepeltier \emph{et al.}~\cite{LMX2005}. 
\begin{lemma}\label{uniqueness}
	If $(V^{\lambda}, M^{\lambda}, A^{\lambda}) \in \cS^2\times \cH^2\times \cK^2$ is a solution to the reflected BSDE in \eqref{entropyrbsde}, then it is unique.
\end{lemma}

\begin{remark}\label{remark4.1}
	Let us point our that the RBSDE in \eqref{entropyrbsde} lies outside the scope of the existing frameworks developed by Zheng \cite{Z2024} and Zheng \emph{et al.} \cite{ZZM2024} for singular reflected BSDEs. Unlike the settings considered in \cite{Z2024, ZZM2024}, this formulation lacks a quadratic term accompanying the singular driver $\ln(1/y)$ for $y \ge 0$. The presence of such a term is a crucial assumption in the analysis of the aforementioned works, as their existence results rely on a transformation technique where the quadratic term naturally arises via It\^o's formula and the domination method (see Bahlali \cite{B2019}, Bahlali \emph{et al.} \cite{BEO2017}, and Bahlali and Tangpi \cite{BT2021}). 
	
\end{remark}

As already mentioned, our analysis relies on a monotone stability argument for BSDEs, in the spirit of Peng~\cite{P1999}, adapted to the entropy-regularized structure of the driver. More precisely, we decompose the generator $\Phi_{\lambda,n}$ into its positive and negative parts, denoted by $\Phi^+_{\lambda,n}$ and $\Phi^-_{\lambda,n}$. The positive component plays a role analogous to the penalization term in classical schemes for reflected BSDEs, while the negative component is shown to be uniformly Lipschitz, with a constant independent of $n$.
To handle the logarithmic singularity at the lower barrier $P$ in the limit $n \to \infty$, we additionally introduce a suitable $\varepsilon$-truncation of the driver. 

To help the reader visualize the driver and the $\varepsilon$-truncation of the driver, we refer to Figure \ref{fig:phi1n} and Figure \ref{fig:phitrun} below
\vskip5pt
\begin{figure}[thp]
	\centering
	
	\begin{minipage}[t]{0.48\textwidth}
		\centering
		\begin{tikzpicture}[scale=0.82]
			\begin{axis}[
				legend style={
					at={(0.66, 0.99)},
					anchor=north west,
					legend columns=1,
					column sep=1ex,
					draw=black,
					inner sep=1ex,
					cells={anchor=west}
				},
				xmax=5, ymax=8, samples=65,
				xticklabel=\empty,
				]
				\addplot[black,thick] {ln( (exp(1*(1-x)) -1)/(1-x) )};
				\addplot[purple,thick] {ln( (exp(2*(1-x)) -1)/(1-x) )};
				\addplot[green,thick] {ln( (exp(5*(1-x)) -1)/(1-x) )};
				\addplot[red,thick] {ln( (exp(15*(1-x)) -1)/(1-x) )};
				\addplot[blue,thick] {-ln( x-1 )};
				\legend{$n=1$,$n=2$,$n=5$,$n=15$, $n = \infty$}
				\addplot[black,dotted] {0};
				\addplot[color=blue, thick] coordinates {(1.1, 2.302) (1.01, 4)};
				\addplot[color=blue, thick] coordinates {(1.01, 4) (1, 7)};
				\addplot[color=black, dotted] coordinates {(1, -1.8) (1, 15)};
				\node[] at (axis cs: 1, -2.2) {{\footnotesize $P_s$}};
				\addplot[color=black, dotted] coordinates {(2, -1.8) (2, 15)};
				\node[] at (axis cs: 2.5, -2.2) {{\footnotesize $P_s + \lambda$}};
				\node[] at (axis cs: -1.9, 7) {{ $\Phi_{1,n}(s,x)$}};
			\end{axis}
		\end{tikzpicture}
		\caption{Sketches of $\Phi_{1,n}$ for various $n$ and of $\Phi_{1,\infty}(s,x)
			= -\ln(x-P_s)\,\mathbf{1}_{\{ x > P_s\}} + \infty\,\mathbf{1}_{\{x \leq P_s\}}$.}
		\label{fig:phi1n}
	\end{minipage}
	\hfill
	\begin{minipage}[t]{0.48\textwidth}
		\centering
		\begin{tikzpicture}[scale=0.82]
			\begin{axis}[
				legend style={
					at={(0.63, 0.99)},
					anchor=north west,
					legend columns=1,
					column sep=1ex,
					draw=black,
					inner sep=1ex,
					cells={anchor=west}
				},
				xmax=5, ymax=8, samples=65,
				xticklabel=\empty,
				]
				\addplot[purple,thick] {ln( (exp(100*(1-max(x,1.5))) -1)/(1-max(x,1.5)) )};
				\addplot[green,thick]  {ln( (exp(100*(1-max(x,1.2))) -1)/(1-max(x,1.2)) )};
				\addplot[black,thick]  {ln( (exp(100*(1-max(x,1.01))) -1)/(1-max(x,1.01)) )};
				\legend{$\epsilon=0.5$,$\epsilon=0.2$,$\epsilon=0.01$}
				\addplot[black,dotted] {0};
				\addplot[color=black, dotted] coordinates {(1, -1.8) (1, 15)};
				\node[] at (axis cs: 1, -2.2) {{\footnotesize $P_s$}};
				\addplot[color=black, dotted] coordinates {(2.3, -1.8) (2.3, 15)};
				\node[] at (axis cs: 2.5, -2.2) {{\footnotesize $P_s + \lambda$}};
				\node[] at (axis cs: -2, 7) {{ $\Phi_{\lambda,n}(s,x\vee(P_s + \epsilon))$}};
			\end{axis}
		\end{tikzpicture}
		\caption{Illustrative sketch of the truncated generator $\Phi_{\lambda,n}(s,x\vee (P_s+\epsilon))$ for several values of $\epsilon$. The functions are Lipschitz continuous and increase as $\epsilon \to 0$.}
		\label{fig:phitrun}
	\end{minipage}
	
\end{figure}

\subsection{Auxiliary Lemmas and Estimates}
This subsection collects technical lemmas and quantitative estimates used in
subsections \ref{subsec:4.2} and \ref{sss3.2.3}. They serve as \emph{a priori} bounds for the
entropy-regularized penalization scheme and will be applied repeatedly to
control the generator and its derivatives uniformly in the penalization
parameters. Specfically, we establish auxiliary properties of the generator
$\Phi_{\lambda,n}(s,x)$, together with integrability and growth estimates
for its time-integrated form. These results are key to proving the
asymptotic convergence to the American option value and to identifying
the limit $V^\lambda$ as the solution of a singular reflected BSDE.

For fixed $\lambda\in(0,1]$, the map $x\mapsto\Phi_{\lambda,n}(s,x)$ is
decreasing and has a unique root at $P_s-\lambda\Phi^{-1}_n(0)$.
\begin{lemma} \label{lem:phin_root}
	For all $n \geq 1$, one has $\Phi_{n}^{-1}(0) \in (-1, 0]$ and $\Phi^{-1}_n(0) \downarrow -1$ as $n \rightarrow \infty$.
\end{lemma}

\begin{proof}
	The case $n=1$ yields $\Phi^{-1}_1(0)=0$.  
	For \(n\ge2\), the root of $\Phi_n$ coincides with the non-zero solution of $f(x) = e^{nx} -x - 1$. Since $f'(x) = ne^{nx} - 1$ and $f''(x) = n^{2}e^{nx}$, the unique minimizer is $x_* = -\frac{1}{n}\ln n \in (-1, 0)$ and $f(x_*) = \frac{1}{n} + \frac{1}{n}\ln n - 1 < 0$. Moreover, \(f(-1)=e^{-n}>0\), and hence the intermediate value theorem 
	implies the existence of a unique root $\Phi_{n}^{-1}(0) \in (-1,x_*)\subset(-1,0)$.
	Differentiating the identity $e^{n\Phi^{-1}_n(0)}-\Phi^{-1}_n(0)-1=0$ with 
	respect to $n$ yields
	$$
	\frac{d}{dn}[\Phi^{-1}_n(0)]
	=\frac{\Phi^{-1}_n(0)e^{n\Phi^{-1}_n(0)}}{1-ne^{n\Phi^{-1}_n(0)}}<0.
	$$
	\noindent Since $|\Phi^{-1}_n(0)+1|=e^{n\Phi^{-1}_n(0)}\le e^{n x_{*}} = \frac{1}{n}$, we conclude that $\lim_{n\rightarrow \infty}\Phi^{-1}_n(0) = -1$.   
\end{proof}

As a direct consequence of Lemma \ref{lem:phin_root}, the root of 
$x\mapsto\Phi_{\lambda,n}(s,x)$ increases to $P_s+\lambda$ as 
$n\to\infty$, and satisfies
\begin{equation}\label{ineq:ps:root:phin}
	P_s \le P_s-\lambda\Phi^{-1}_n(0)\le P_s+\lambda.
\end{equation}

\begin{lemma}\label{lip}
	Fix \(n\ge1\).  
	The derivative $x\mapsto\partial_x\Phi_{\lambda,n}(s,x)$ is negative and 
	increasing.  
	Furthermore, for any $x_*>0$ and all $x\ge P_s+x_*$,
	$$
	\partial_x\Phi_{\lambda,n}(s,x)
	\;\ge\; \left(-\frac{\lambda}{x_*}\right)\vee\left(-\frac n2\right).
	$$
	\begin{proof}
		From \eqref{def: phi_lambda_n} and Lemma \ref{lipphi} 
		\begin{align*}
			\partial_x \Phi_{\lambda, n}(s,x) & = -n \Phi'\left(\frac{P_{s}-x}{\lambda / n } \right) < 0,\\
			\partial^2_{xx} \Phi_{\lambda, n}(s,x) &  = \frac{n^{2}}{\lambda} \Phi''\left(\frac{P_{s}-x}{\lambda/n} \right) > 0.
		\end{align*}
		Hence the minimum over $[P_s+x_*,\infty)$ is attained at $x=P_s+x_*$.  
		Setting $y=nx_*/\lambda$, we have
		\begin{align}\label{first:ineq:partial:x:phi:lambda:n}
			\partial_x \Phi_{\lambda, n}(s, P_{s}+x_{*}) & = -n \Phi'\left(-nx_{*} / \lambda \right) = \frac{\lambda}{x_*}\frac{{{y}e^{-y}} - (1-e^{-y})}{1-e^{-y}} \geq -\frac{\lambda}{x_*},
		\end{align}
		where we use the fact that for $y\geq 0$
		\begin{equation*}
			\frac{{{y}e^{-y}} - (1-e^{-y})}{1-e^{-y}} = \frac{{{y}e^{-y}}}{1-e^{-y}}  - 1 \in [-1,0).
		\end{equation*}
		Using positivity and monotonicity of \(\Phi'\) gives
		$$
		\partial_x \Phi_{\lambda,n}(s,P_s+x_*) = -n \Phi'(-nx_\star/\lambda)
		\ge -n\Phi'(0)=-\frac n2.
		$$
	\end{proof}
\end{lemma}
\begin{corollary}\label{cor3.1}
	For all $x \in [P_s - \lambda\Phi^{-1}_n(0), \infty)$, 
	$$
	|\partial_x\Phi_{\lambda,n}(s,x)|  
	\leq \max\{1,1/|\Phi^{-1}_2(0)|\}.
	$$
\end{corollary}
\begin{proof}
	By Lemma \ref{lem:phin_root}, we have $x_n := n\Phi_n^{-1}(0) \in (-n, 0]$.
	The result then follows from \eqref{first:ineq:partial:x:phi:lambda:n} since, for $n=1$, one has $\partial_x \Phi_{\lambda, 1}(s, P_{s} - \lambda x_{1}) = - \Phi'(0) = -\frac{1}{2}$ and $|\Phi^{-1}_n(0)| \leq |\Phi^{-1}_{n+1}(0)|$ for $n\geq 2$.
\end{proof}

We now provide some a priori estimates on \eqref{vn}. To proceed, we set 
\begin{align}
	K^{\lambda,n}_t & := \int^t_0 \Phi_{\lambda,n}(s, V^{\lambda,n}_s)\, ds, \label{Kln}\\
	K^{\lambda,n,\pm}_t & := \int^t_0 \Phi^{\pm}_{\lambda,n}(s, V^{\lambda,n}_s)\, ds.\label{Kpm}
\end{align}

\begin{lemma}\label{VMKestimate}
	Assume that Assumption \ref{A} holds. Let $(V^{\lambda,n}, M^{\lambda,n})\in \cS^2\times \cH^2$ denote the unique solution to \eqref{vn}. Then there exists a constant $C<\infty$ such that, for any $n\geq1$ and any $0\leq \lambda \leq 1$
	\begin{align*}
		\EE[\sup_{0\leq t\leq T} |V^{\lambda,n}_t|^2] + \EE\left[[M^{\lambda,n}]_T\right]  + \EE\big[ |K^{\lambda,n,+}_T|^2\big] + \EE\big[ |K^{\lambda,n,-}_T|^2\big]\leq C.
	\end{align*}
\end{lemma}
\begin{proof}
	
	\noindent \emph{Step 1.} We write $\Phi_{\lambda,n} = \Phi_{\lambda,n}^+ - \Phi_{\lambda,n}^{-}$ and observe by \eqref{ineq:ps:root:phin} that 
	\begin{equation} \label{support:phin:positive:part}
		\left\{ x: \Phi_n\left(\lambda^{-1}(P_s-x)\right) \geq  0 \right\} = \left\{ x: P_s - \lambda\Phi_n^{-1}(0)  \geq x \right\}  \subseteq \left\{ x: P_s + \lambda \geq x \right\},
	\end{equation}
	
	\noindent and similarly, 
	\begin{equation}\label{support:phin:negative:part}
		\left\{ x: \Phi_n\left(\lambda^{-1}(P_s-x)\right) < 0 \right\} = \left\{ x: P_s - \lambda\Phi_n^{-1}(0)  <  x \right\}  \subseteq \left\{ x: P_s < x \right\}.
	\end{equation}
	
	From \eqref{support:phin:positive:part}, the positive part of the generator given by \eqref{def: phi_lambda_n} satisfies for all $x\in \mathbb{R}$
	\begin{align}\label{ineq:positive:part:Philambdan}
		x\Phi^{+}_{\lambda,n}(s, x) \leq (P_s+ \lambda)\Phi^{+}_{\lambda,n}(s, x) , \quad a.s.
	\end{align}
	By \eqref{ineq:positive:part:Philambdan} and Young's inequality, we obtain for any $\alpha>0$:
	\begin{equation}\label{first:estimate:vlambda:n}
		\begin{aligned}
			& (V^{\lambda,n}_t)^2 + \int^T_t d[M^{\lambda,n}]_s\\
			& = P^2_T - 2\int^T_t V^{\lambda,n}_sdM^{\lambda,n}_s  + 2\int^T_t V^{\lambda,n}_s  \Phi^+_{\lambda,n}(s, V^{\lambda,n}_s) ds - 2\int^T_t V^{\lambda,n}_s \Phi^-_{\lambda,n}(s, V^{\lambda,n}_s) ds. \\
			& \leq  P^2_T - 2\int^T_t V^{\lambda,n}_sdM^{\lambda,n}_s + 2\int^T_t (P_s+\lambda) \Phi^+_{\lambda,n}(s, V^{\lambda,n}_s)  ds \\
			& \leq  P^2_T - 2\int^T_t V^{\lambda,n}_sdM^{\lambda,n}_s + \frac{2}{\alpha}\sup_{t< s\leq T } (P_s +\lambda)^2 +  2\alpha\big(K^{\lambda,n,+}_T - K^{\lambda,n,+}_t\big)^2. 
		\end{aligned}
	\end{equation}
	
	\emph{Step 2.} Using \eqref{Kpm}, the dynamics \eqref{vn} writes
	\begin{align}
		V^{\lambda,n}_t 
		& = P_T - \int^T_t dM^{\lambda,n}_s + (K^{\lambda,n,+}_T - K^{\lambda,n,+}_t) - (K^{\lambda,n,-}_T - K^{\lambda,n,-}_t).
	\end{align}
	By Jensen’s inequality, Itô’s isometry, and the Lipschitz property of $\Phi^-_{\lambda,n}$ (with constant $c$ independent of $n$ and decreasing in $\lambda$, see Corollary \ref{cor3.1}), \eqref{lem:phin_root} and noting that $\Phi^-_{\lambda,n}(s,P_s-\lambda \Phi_n^{-1}(0)) = 0$, we obtain
	\begin{align*}
		& \EE\big[\big(K^{\lambda,n,+}_T - K^{\lambda,n,+}_t\big)^2\big]\\
		& \leq {4} \EE\Big[P_T^2 + (V^{\lambda,n}_t)^2 + [M^{\lambda,n}]_T - [M^{\lambda,n}]_t + \Big(\int^T_t \Phi_{\lambda,n}^-(s,V^{\lambda,n}_s) ds \Big)^2\Big] \\
		& \leq {4}  \EE\Big[P_T^2 + (V^{\lambda,n}_t)^2 + [M^{\lambda,n}]_T - [M^{\lambda,n}]_t + c\Big(\int^T_t |V^{\lambda,n}_s - P_s + \lambda \Phi^{-1}_{\lambda,n}(0)| ds \Big)^2\Big] \\
		& \leq {4(3cT\vee 1)} \EE\Big[P_T^2 + |V^{\lambda,n}_t|^2+ [M^{\lambda,n}]_T - [M^{\lambda,n}]_t + \int^T_t ((V^{\lambda,n}_s)^2 + P_s^2 + \lambda^2 ) ds \Big].
	\end{align*}
	Then by plugging the above estimates into \eqref{first:estimate:vlambda:n} and picking $2\alpha = 1/(3 ({4(3cT\vee 1)}))$ we obtain that 
	\begin{align}
		\frac{2}{3}\EE[(V^{\lambda,n}_t)^2] + \frac{2}{3}\EE\big[[M^{\lambda,n}]_T - [M^{\lambda,n}]_t\big] \leq C(\lambda)\Big( 1+  \int^T_t \EE[(V^{\lambda,n}_s)^2]ds \Big), \label{SupV}
	\end{align}
	where the positive constant $C(\lambda)$ depends on $\lambda$ in an increasing way. Hence, by applying Gr\"onwall's inequality to \eqref{SupV}, we obtain $\sup_{0\leq t\leq T}\EE[|V^{\lambda,n}_t|^2] \leq C(\lambda)e^{TC(\lambda)} \leq C(1)e^{TC(1)}$ which in turn yields
	\begin{align}
		\sup_{0\leq t\leq T}\EE[(V^{\lambda,n}_t)^2] + \EE\left[[M^{\lambda,n}]_T\right] + \EE\big[\big(K^{\lambda,n,+}_T \big)^2\big] \leq C.\label{Clambda}
	\end{align}
	
	\vskip5pt
	\emph{Step 3.} By \eqref{first:estimate:vlambda:n}, we obtain 
	\begin{align}
		\sup_{0\leq t\leq T}(V^{\lambda,n}_t)^2 + \int^T_0 d[M^{\lambda,n}]_s
		& \leq  P^2_T + 2\sup_{0<t \leq T} \Big|\int^T_t V^{\lambda,n}_sdM^{\lambda,n}_s \Big|
		+ \frac{2}{\alpha}\sup_{0\leq t \leq T } (P_t +\lambda)^2 \nonumber \\
		& \qquad +  2\alpha\Big(\int^T_0 \Phi^+_{\lambda,n}(s, V^{\lambda,n}_s)  ds\Big)^2.  \label{supv2}
	\end{align}
	By the Burkholder-Davis-Gundy and the Young inequalities, we have
	\begin{align}
		2\EE\Big[\,\sup_{0\leq t \leq T}\Big|\int^T_t V^{\lambda,n}_sdM^{\lambda,n}_s\Big|\Big] & \leq 2c_1\EE\Big[\Big(\int^T_0 |V^{\lambda,n}_s|^2d[M^{\lambda,n}]_s\Big)^\frac{1}{2}\Big]  \nonumber \\
		& \leq \frac{1}{2}\EE[\,\sup_{0\leq t\leq T}(V^{\lambda,n}_t)^2] + 2c_1^2\EE[[M^{\lambda,n}]_T]. \label{BDG}
	\end{align}
	By substituting the inequality \eqref{BDG} and \eqref{Clambda} into \eqref{supv2} and again by picking $2\alpha = 1/(3 ({4(3cT\vee 1)}))$, we obtain $\EE[\sup_{0\leq t\leq T}	(V^{\lambda,n}_t)^2] \leq C$.

	Lastly, since the map $x\mapsto\Phi^-_{\lambda,n}(s,x)$ is Lipschitz continuous with a Lipschitz constant independent of the penalty parameter $n$ (see Corollary \ref{cor3.1}), we deduce that there exists a positive constant $C$, independent of $n$, such that $\EE\big[ (K^{\lambda,n,-}_T)^2\big]\leq C(\lambda)$.
\end{proof}

\begin{corollary} \label{cor: VlS2}
	For any $\lambda \in [0,1]$,  $V^{\lambda} \in \cS^{2}$.
\end{corollary}
\begin{proof}
	The result follows from Lemma \ref{VMKestimate}, \eqref{existence:limit:Vlambda} and Fatou's lemma.
\end{proof}

In addition to the terms in \eqref{Kln} and \eqref{Kpm}, for a fixed $\epsilon \in (0, -\lambda\Phi^{-1}_2(0))$, we introduce an $\epsilon$-truncation of the driver. By performing this truncation step, we control the derivative of the driver $\Phi_{\lambda,n}$ using $\epsilon$ and prevent it from exploding as $n\rightarrow \infty$. This strategy enables us to first pass the limit as $n\rightarrow \infty$ first and then to use monotone arguments to pass the limit as $\epsilon\rightarrow 0$. Specifically, we consider
\begin{align}
	& K^{\lambda,n,\epsilon}_t  := \int^t_0 \Phi_{\lambda,n}(s, V^{\lambda,n}_s\vee (P_s+ \epsilon)) \, ds, \label{KLE}\\
	&  K^{\lambda,n,\epsilon,\pm}_t := \int^t_0 \Phi^\pm_{\lambda,n}(s, V^{\lambda,n}_s\vee (P_s+ \epsilon)) \, ds, \\
	& A^{\lambda,n,\epsilon}_t  := \int^t_0 [\Phi_{\lambda,n}(s, V^{\lambda,n}_s)- \Phi_{\lambda,n}(s, V^{\lambda,n}_s\vee (P_s+ \epsilon))] \, ds. \label{ALE}
\end{align}
We also define the two limiting processes
\begin{align}
	& K^{\lambda,\epsilon}_t  := \lim_{n\rightarrow \infty} K^{\lambda,n,\epsilon}_t  = \int^t_0 \Phi_{\lambda, \infty}(s, V^{\lambda}_s\vee (P_s+ \epsilon)) \, ds,\label{eq3.24} \\ 
	& K^{\lambda}_t :=\lim_{\epsilon\downarrow 0} K^{\lambda,\epsilon}_t = \int^t_0 \lambda\Phi_{\lambda, \infty}(s, V^{\lambda}_s\vee P_s) \, ds.\label{eq3.25}
\end{align}

To justify the limit in \eqref{eq3.24}, recall from Corollary \ref{cor3.1} that the driver $x \mapsto \Phi_{\lambda,n}(s, x\vee (P_s+ \epsilon))$ is Lipschitz continuous with Lipschitz constant $\lambda/\epsilon$. This yields the bound on the integrand in \eqref{KLE}
$$|\Phi_{\lambda,n} (s, V_s^{\lambda,n}\vee (P_s+ \epsilon))| \leq \frac{\lambda}{\epsilon}|V^{\lambda,n}_s - P_s+\lambda \Phi^{-1}_n(0)| \leq \frac{\lambda}{\epsilon}\left({|V^{\lambda}_s|}+|P_s|+\lambda \right)
$$ 
\noindent which is integrable since $V^{\lambda}$ and $P$ are elements of $\cS^{2}$. By Lemma \ref{lip} we have
\begin{align}
	& \big|\Phi_{\lambda,n}(s, V^{\lambda,n}_s\vee (P_s+ \epsilon)) -\Phi_{\lambda, \infty}(s, V^{\lambda}_s\vee (P_s+ \epsilon))\big| \nonumber \\
	& \leq \big|\Phi_{\lambda,n}(s, V^{\lambda,n}_s\vee (P_s+ \epsilon)) - \Phi_{\lambda,n}(s, V^{\lambda}_s\vee (P_s+ \epsilon))| \nonumber \\
	& \quad + |\Phi_{\lambda,n} (s, V^{\lambda}_s\vee (P_s+ \epsilon)) -\Phi_{\lambda, \infty}(s, V^{\lambda}_s\vee (P_s+ \epsilon))\big| \nonumber \\
	& \leq \frac{1}{\epsilon}|V^{\lambda,n}_s - V^{\lambda}_s| +  \big|\Phi_{\lambda,n} (s, V^{\lambda}_s\vee (P_s+ \epsilon)) -\Phi_{\lambda, \infty}(s, V^{\lambda}_s\vee (P_s+ \epsilon))\big|. \label{klambdan}
\end{align}

By \eqref{def: phi_lambda_n}, \eqref{deriv:phin} and \eqref{increasing:sequence:vlambdan}, the two terms above on the right-hand side converge monotonically to zero as $n\rightarrow \infty$. The limit in \eqref{eq3.25} stems from the fact that $\epsilon \mapsto \Phi_{\lambda, \infty}(s, x \vee (P_{s} + \epsilon))$ is monotonically decreasing combined with the monotone convergence theorem.

We now give some uniform estimates in the $\cS^2$ norm for the processes $K^{\lambda,n,\epsilon}$, $K^{\lambda,n,\epsilon,\pm}$, $A^{\lambda,n,\epsilon}$, $K^{\lambda, \epsilon}$ and $K^\lambda$.

\begin{lemma} \label{lem: Ale_Klne_s2_bound}
	There exists a constant $C<\infty$ such that for any $0 < \lambda \leq 1$, any $n\geq 2$ and any $\epsilon \in (0,- \lambda\Phi^{-1}_2(0))$,
	\begin{gather*}
		\|A^{\lambda, n ,\epsilon}\|_{\cS^2} +  \|K^{\lambda,n,\epsilon,+}\|_{\cS^2}\leq C.
	\end{gather*}
\end{lemma}
\begin{proof}
	We first note that $A^{\lambda, n ,\epsilon} + K^{\lambda,n,\epsilon,+} = K^{\lambda,n,+}$. Indeed, from \eqref{support:phin:positive:part} and \eqref{support:phin:negative:part} the support of $\Phi_{\lambda,n}^+(s, \cdot)$ is 
	\begin{equation*}
		\left\{ x: \Phi_n\left(\lambda^{-1}(P_s-x)\right) \geq  0 \right\} = \left\{ x: P_s - \lambda\Phi_n^{-1}(0)  \geq x \right\},
	\end{equation*}
	and the support of $\Phi_{\lambda,n}^{-}(s, \cdot)$ is
	\begin{equation*}
		\left\{ x: \Phi_n\left(\lambda^{-1}(P_s-x)\right) < 0 \right\} = \left\{ x: P_s - \lambda\Phi_n^{-1}(0)  <  x \right\}.
	\end{equation*}
	But since $\epsilon < - \lambda \Phi_{2}^{-1}(0) \leq -\lambda\Phi_{n}^{-1}(0)$, we have
	\begin{align*}
		A^{\lambda, n, \epsilon}_{t} & + K^{\lambda, n, \epsilon, +}_{t} \! =\! \int^t_0 \left[ \Phi_{\lambda,n}(s, V^{\lambda,n}_s)\!-\! \Phi_{\lambda,n}(s, V^{\lambda,n}_s\vee (P_s+ \epsilon))  \right] \I_{\{V_{s}^{\lambda, n} \leq P_{s} - \Phi_{n}^{-1}(0)\}} ds  \\
		&\hspace{1em} + \int^t_0 \left[ \Phi_{\lambda,n}(s, V^{\lambda,n}_s)- \Phi_{\lambda,n}(s, V^{\lambda,n}_s\vee (P_s+ \epsilon))  \right] \I_{\{V_{s}^{\lambda, n} > P_{s} - \Phi_{n}^{-1}(0)\}} ds  \\
		&\hspace{1em} + \int_{0}^{t}\Phi^{+}_{\lambda,n}(s, V^{\lambda,n}_s\vee (P_s+ \epsilon)) ds = K^{\lambda,n,+}_{t}.
	\end{align*}
	On the other hand, $A^{\lambda, n ,\epsilon}$, $K^{\lambda,n,\epsilon,+} \geq 0$, and the result follows from Lemma \ref{VMKestimate}.
\end{proof}

\begin{lemma} \label{lem: KleS2}
	There exists a constant $C<\infty$ such that, for any $0 < \lambda \leq 1$ and any $\epsilon \in (0,-\lambda\Phi^{-1}_2(0))$
	$$
	\|K^{\lambda}\|_{\cS^2} + \|K^{\lambda,\epsilon}\|_{\cS^2} \leq C.
	$$
\end{lemma}

\begin{proof}
	Fix $\epsilon < -\lambda\Phi^{-1}_2(0)$ and $n\geq 2$. We first note that
	\begin{align*}
		K^{\lambda,n,\epsilon,\pm}_t & = \int^t_0 \Phi^{\pm}_{\lambda,n}(s,V^{\lambda,n}_s\vee (P_s+\epsilon)) ds \\
		& \leq \int^t_0 \Phi^{\pm}_{\lambda,n}(s,V^{\lambda,n}_s) ds = K^{\lambda,n,\pm}_t.
	\end{align*}
	
	Hence, by Lemma \ref{VMKestimate} and Fatou’s lemma,
	$$
	\EE[|K_T^{\lambda,\epsilon,\pm}|^2] \leq \liminf_{n\rightarrow \infty}\EE[|K_T^{\lambda,n,\epsilon,\pm}|^2] \leq \liminf_{n\rightarrow \infty}
	\EE[|K^{\lambda,n,+}_T|^2] \leq C. 
	$$
	
	Moreover, $\sup_{0\leq t\leq T}|K^{\lambda,\epsilon}_t|^2 \leq  2|K^{\lambda,\epsilon,+}_T|^2 +  2|K^{\lambda,\epsilon,-}_T|^2$ and it follows that $K^{\lambda, \epsilon} \in \cS^2$ with $\|K^{\lambda,\epsilon}\|_{\cS^2}$ bounded above by a constant $C$ that is independent of $\epsilon$ and $\lambda$. 
	
	To obtain an upper bound for $\|K^{\lambda}\|_{\cS^2}$, observe that 
	$$
	|K^{\lambda}_t|^2 = \liminf_{\epsilon\downarrow 0 }|K^{\lambda,\epsilon}_t|^2 \leq \liminf_{\epsilon\downarrow 0 }\sup_{0\leq t\leq T} |K^{\lambda,\epsilon}_t|^2.
	$$
	By Fatou’s lemma and the bound above, we have $\|K^\lambda\|_{\cS^2} \leq \liminf_{\epsilon\downarrow 0} \|K^{\lambda,\epsilon}\|_{\cS^2} \leq C$.
\end{proof}

\subsection{Asymptotic Convergence to the American Option} \label{subsec:4.2} 
In this subsection, we establish a convergence rate for $V^\lambda$ towards the American option value $V$ as $\lambda\downarrow0$. Importantly, we point out that here $V^\lambda$ is only considered as the monotone limit of $V^{\lambda,n}$, whose existence follows from the standard monotone convergence arguments presented at the beginning of this section.

The reflected BSDE characterization of $V^\lambda$ is deferred to subsection \ref{sss3.2.3}, where,
under stronger assumptions, we show that $V^\lambda$ arises as the value component of a reflected BSDE with a singular driver and admits an optimal stopping representation.
Recall that, since the payoff process $P$ is càdlàg and belongs to $\mathcal{S}^2$, the Doob-Meyer decomposition applied to the supermartingale $V$ guarantees the existence of a martingale $M$ and a continuous, non-decreasing process $A$ such that $V = M - A$. Moreover, the process $A$ satisfies the Skorokhod reflection condition $\int^T_0 (P_s - V_s)dA_s = 0$ which characterises the minimality of the reflection and ensures that $V$ coincides with the value of the American option.

We begin by establishing an estimate for the difference between $V^{\lambda,n}$ and $V$.

\begin{lemma} \label{lem: estVlnV}
	For any $0\leq \lambda\leq 1$ and any $n\geq 1$,
	\begin{align}
		e^{t}(V^{\lambda,n}_t - V_t)^2  & \leq 2\lambda e^{T}\mathbb{E}[K^{\lambda,n,+}_T-K^{\lambda,n,+}_t|\cF_t] \nonumber\\
		& \quad + 2e^{T}\mathbb{E}\Big[\int^T_t (V^{\lambda,n}_s - P_s)^{-}dA_s| \cF_t\Big] \label{yny}\\
		& \quad + e^{T}\mathbb{E}\Big[\int^T_t [\Phi^-_{\lambda,n}(s,V_s)]^2 ds | \cF_t\Big]. \nonumber 
	\end{align}
\end{lemma}
\begin{proof}
	Applying Itô’s formula and using the reflected BSDE representation of $V$ in \eqref{VRBSDE}, we obtain
	\begin{align*}
		& e^{\beta t }(V^{\lambda,n}_t - V_t)^2 + \int^T_te^{\beta s} d[\bar M^{\lambda,n}]_s\\
		& = - 2\int^T_te^{\beta s} (V^{\lambda, n}_{s} - V_{s}) d\bar M^{\lambda,n}_s+ 2\int^T_t e^{\beta s }(V^{\lambda,n}_s - V_s)  d(K^{\lambda,n,+}_s - A_s)\\
		& \quad  -2 \int^T_t e^{\beta s}(V^{\lambda,n}_s - V_s)\Phi^-_{\lambda,n}(s,V^{\lambda,n}_s) ds - \beta\int^T_t e^{\beta s}(V^{\lambda,n}_s - V_s)^2 ds,
	\end{align*}
	where $\bar M^{\lambda,n} = M^{\lambda,n} - M$. We first estimate the reflection term. Observe that
	\begin{align*}
		& \int^T_t e^{\beta s} (V^{\lambda,n}_s - V_s)  d(K^{\lambda,n,+}_s - A_s) \\
		& = \int^T_t e^{\beta s} (V^{\lambda,n}_s - P_s - \lambda)dK^{\lambda,n,+}_s  + \int^T_t e^{\beta s} (P_s + \lambda- V_s)dK^{\lambda,n,+}_s \\
		& \quad -\int^T_t e^{\beta s} (V^{\lambda,n}_s - P_s)dA_s  + \int^T_t e^{\beta s}(V_s- P_s)dA_s\\
		& \leq \lambda e^{\beta T} (K^{\lambda,n,+}_T-K^{\lambda,n,+}_t)  + e^{\beta T} \int^T_t(V^{\lambda,n}_s - P_s)^{-}dA_s,
	\end{align*}
	
	\noindent where the inequality follows from the facts that $\int^t_0 (V_s-P_s)dA_s = 0$, $V \geq P$ and that the non-decreasing process $K^{\lambda, n}$ increases only when $V^{\lambda,n} - P- \lambda \leq 0$. 
	
	Next, we estimate the generator. We write 
	\begin{align*}
		& \int^T_t e^{\beta s} (V^{\lambda,n}_s - V_{s})[-\Phi^-_{\lambda,n}(s,V^{\lambda,n}_s)] ds\\
		& = \int^T_t e^{\beta s}(V^{\lambda,n}_s - V_s)[-\Phi^-_{\lambda,n}(s,V^{\lambda,n}_s) + \Phi^-_{\lambda,n}(s,V_s)]ds \\
		& \quad + \int^T_t e^{\beta s} (V^{\lambda,n}_s - V_s)[-\Phi^-_{\lambda,n}(s, V_s)]ds \\
		& \leq  \int^T_t e^{\beta s} (V^{\lambda,n}_s - V_s)^+\Phi^-_{\lambda,n}(s,V_s)ds, 
	\end{align*}
	where the inequality follows from the fact that $x\mapsto -\Phi_{\lambda,n}^{-}(s,x)$ is a decreasing function. Combining the above estimates and applying Young’s inequality yields
	\begin{align*}
		&  e^{\beta t}(V_{t}^{\lambda,n} - V_{t})^2  \leq 2\lambda e^{\beta T}\mathbb{E}[K^{\lambda,n,+}_T-K^{\lambda,n,+}_t|\cF_t] - \EE \Big[\beta\int^T_t e^{\beta s}(V^{\lambda,n}_s - V_s)^2 ds | \cF_{t} \Big]\\
		& \quad + 2e^{\beta T}\mathbb{E}\Big[\int^T_t (V^{\lambda,n}_s - P_s)^{-}dA_s| \cF_t\Big] + 2\mathbb{E}\Big[\int^T_t e^{\beta s}(V^{\lambda,n}_s - V_s)^+\Phi^-_{\lambda,n}(s,V_s) ds | \cF_t\Big] \\
		& \leq 2\lambda e^{\beta T}\mathbb{E}[K^{\lambda,n,+}_T-K^{\lambda,n,+}_t|\cF_t] - \EE\Big[\beta\int^T_t e^{\beta s}(V^{\lambda,n}_s - V_s)^2 ds | \cF_{t} \Big]\\
		& \quad + 2e^{\beta T}\mathbb{E}\Big[\int^T_t (V^{\lambda,n}_s - P_s)^{-}dA_s| \cF_t\Big]\\
		& \quad + 2\mathbb{E}\Big[\int^T_t e^{\beta s}\frac{1}{2}(V^{\lambda,n}_s - V_s)^2 + e^{\beta s}\frac{1}{2}[\Phi^-_{\lambda,n}(s,V_s)]^2 ds | \cF_t \Big].
	\end{align*}
	Choosing $\beta = 1$ yields the desired estimate \eqref{yny}.
\end{proof}


\begin{theorem}\label{t4.1}
	There exists a constant $C<\infty$ such that for any $0\leq t\leq T$ and any $0<\lambda \leq 1$
	\begin{align*}
		(V^{\lambda}_t - V_t)^2
		& \leq C\Big(\lambda \liminf_{n\rightarrow \infty} \mathbb{E}[K^{\lambda,n,+}_T \big| \, \cF_t] + [\lambda \ln(\lambda)]^2 +  \lambda^2\mathbb{E}\big[\sup_{t<s\leq T} V_s^2 \,\big|\, \cF_t\big]\Big). 
	\end{align*}
	If the payoff process $P$ is bounded, then
	\begin{align*}
		(V^{\lambda}_0 - V_0)^2 & \leq C\Big(\lambda + [\lambda \ln(\lambda)]^2 +  \lambda^2 \Big). 
	\end{align*}
\end{theorem}

\begin{proof}
	Taking the lower limit in the estimate of Lemma \ref{lem: estVlnV} yields
	\begin{align*}
		e^{t}(V_{t}^{\lambda} - V_{t})^2 
		&\leq 2\lambda e^{T} \liminf_{n\rightarrow\infty}\mathbb{E}[K^{\lambda,n,+}_T|\cF_t] +2e^{T} \liminf_{n \rightarrow \infty} \mathbb{E}[\int^T_t (V^{\lambda,n}_s - P_s)^{-}dK_s| \cF_t] \\
		&\hspace{1em} + \liminf_{n \rightarrow \infty}\mathbb{E}[\int^T_t e^{s}[\Phi^-_{\lambda,n}(s,V_s)]^2 ds | \cF_t].
	\end{align*}
	
	By Lemma \ref{lem: VlgP} and dominated convergence,
	\begin{align*}
		e^{t}(V_{t}^{\lambda} - V_{t})^2 
		& \leq 2\lambda e^{T} \liminf_{n\rightarrow\infty} \mathbb{E}[K^{\lambda,n,+}_T|\cF_t]  \\
		& \quad + \mathbb{E}[\int^T_t e^{s} \left[\lambda\ln\left(\lambda^{-1}(V_s - P_s)\right)\right	]^2\I_{\{V_s \geq P_s + \lambda\}} ds | \cF_t]. 
	\end{align*}
	
	To proceed, we note that for $V_s -P_s\geq \lambda$
	\begin{align*}
		\lambda\ln\left(\frac{V_s - P_s}{\lambda}\right) 
		&\leq \lambda \ln(V_s) - \lambda \ln (\lambda) \\
		& \leq \lambda V_s- 2\lambda \ln (\lambda),
	\end{align*}
	
	\noindent where the last step uses $|\ln(x)| \leq \max\{-\ln(\lambda), x\}$ for $x\geq \lambda$. This gives
	\begin{align*}
		& e^{t}(V^{\lambda}_t - V_t)^2 \\
		& \leq 2\lambda e^{T}\liminf_{n\rightarrow \infty} \mathbb{E}[K^{\lambda,n,+}_T|\cF_t] + 4(e^{T} - e^{t})(\lambda \ln(\lambda))^2 +  2\lambda^2(e^{T} - e^{t})\mathbb{E}\big[\sup_{t<s\leq T} V_s^2 \,\big|\, \cF_t\big] \\
		& \leq C\Big(\lambda\liminf_{n\rightarrow \infty} \mathbb{E}[K^{\lambda,n,+}_T|\cF_t] + (\lambda \ln(\lambda))^2 +  \lambda^2\mathbb{E}\big[\sup_{t<s\leq T} V_s^2 \,\big|\, \cF_t\big]\Big). 
	\end{align*}

	For $t=0$, we obtain
	\begin{align*}
		(V^{\lambda}_0 - V_0)^2 
		& \leq C\Big( \lambda \liminf_n \mathbb{E}[K^{\lambda,n,+}_T] + (\lambda \ln(\lambda))^2 +  \lambda^2\mathbb{E}\big[\sup_{0\leq s\leq T} V_s^2 \big] \Big).
	\end{align*}
	
	If the payoff $P$ is bounded, then $\sup_{0\leq s\leq T} V_s^2$ is bounded as well, which combined with Lemma \ref{VMKestimate} yields the simplified estimate
	\begin{align*}
		(V^{\lambda}_0 - V_0)^2 
		&\leq C (\lambda + (\lambda \ln(\lambda))^2 +  \lambda^2 ).
	\end{align*}
	This completes the proof.
\end{proof}

\subsection{Reflected BSDE with a Singular Driver}\label{sss3.2.3}
In this subsection, we move beyond the non-asymptotic estimates of
subsection~\ref{subsec:4.2} and investigate the RBSDE satisfied by the limit process
$V^\lambda$. Under additional assumptions, we show that $V^\lambda$, defined
as the monotone limit of $V^{\lambda,n}$, can be characterized as the value
component of a reflected BSDE with a singular driver. In particular, we prove
the existence of processes $M^\lambda \in \mathcal{H}^2$ and
$A^\lambda \in \mathcal{K}^2$ such that
$(V^\lambda, M^\lambda, A^\lambda)$ solves~\eqref{entropyrbsde}. We then
derive an optimal stopping representation for $V^\lambda$, establishing a
direct link between the singular RBSDE and the corresponding stopping
problem. For technical convenience, we impose the following assumption.

\begin{hypothesis} \label{hyp: Fcont}
	All $\FF$ martingales are continuous, and the payoff process $P$ is continuous.
\end{hypothesis}

\begin{remark}
	The above assumption, together with Lemma \ref{lem: VlgP}, ensures that the limit process $V^\lambda$ possesses continuous sample paths. This continuity is essential: it allows us to invoke the Doob-Meyer decomposition to identify the martingale component $M^\lambda$ and the increasing process $A^\lambda$ in the reflected BSDE. 
\end{remark}


\begin{lemma}\label{lem: VlgP} 
	For any $\lambda \in (0,1]$, the entropy-regularized penalization scheme $\{V^{\lambda, n}\}_{n\geq 2}$ is a Cauchy sequence in $\cS^{2}$, and its limit satisfies: for all $t\in [0,T]$, $V^\lambda_t := \lim_{n} V^{\lambda, n}_t  \geq P_t$ a.s.
\end{lemma}

\begin{proof}
	\noindent \emph{Step 1.}
	For any $m, n \geq 2$, It\^o's formula yields 
	\begin{align} \label{ineq: ito_s2_cauchy_1}
		&e^{\beta t }(V^{\lambda,n}_t - V_t^{\lambda,m})^2 \\
		& \leq - 2\int^T_te^{\beta s} (V^{\lambda, n}_{s} - V^{\lambda,m}_{s}) d\bar M^{\lambda,n,m}_s + 2\int^T_t e^{\beta s }(V^{\lambda,n}_s - V_s^{\lambda,m})  d(A^{\lambda,n,\epsilon}_s - A^{\lambda,m,\epsilon}_s) \nonumber \\
		& \quad  + 2\int^T_t e^{\beta s}(V^{\lambda,n}_s - V_s^{\lambda,m})d(K^{\lambda,n,\epsilon,+}_s - K^{\lambda,m,\epsilon,+}_s) \nonumber \\
		& \quad  - 2 \int^T_t e^{\beta s}(V^{\lambda,n}_s - V_s^{\lambda,m})d(K^{\lambda,n,\epsilon,-}_s - K^{\lambda,m,\epsilon,-}_s) - \beta\int^T_t e^{\beta s}(V^{\lambda,n}_s - V_s^{\lambda,m})^2 ds,\nonumber 
	\end{align}
	where $\bar M^{\lambda,n,m} = M^{\lambda,n} - M^{\lambda,m}$. Without loss of generality, we assume $n > m$. We now estimate each term of the above decomposition. By \eqref{ALE}, the process $A^{\lambda, n, \epsilon}$ increases only on the set $\left\{V^{\lambda, n} \leq P + \epsilon\right\}$. Hence,
	\begin{align*}
		& \int^T_t (V^{\lambda,n}_s - V_s^{\lambda,m})  d(A^{\lambda,n,\epsilon}_s - A^{\lambda,m,\epsilon}_s) \\
		& = \int^T_t (V^{\lambda,n}_s - P_s -\epsilon)dA^{\lambda,n,\epsilon}_s  + \int^T_t(P_s + \epsilon - V_s^{\lambda,m})dA^{\lambda,n,\epsilon}_s \nonumber \\
		& \quad  -\int^T_t (V^{\lambda,n}_s - P_s -\epsilon)dA_s^{\lambda,m,\epsilon}  + \int^T_t(V_s^{\lambda,m}- P_s -\epsilon)dA_s^{\lambda,m,\epsilon}\\
		& \leq -\int^T_t(V_s^{\lambda,m} - P_s - \epsilon )dA^{\lambda,n,\epsilon}_s -\int^T_t (V^{\lambda,n}_s - P_s - \epsilon) dA_s^{\lambda,m,\epsilon}\\
		& \leq \int^T_t(V_s^{\lambda,m} - P_s - \epsilon)^{-}  dA^{\lambda,n,\epsilon}_s +\int^T_t (V^{\lambda,n}_s - P_s - \epsilon)^{-} dA_s^{\lambda,m,\epsilon} \\
		&\leq \int^T_0(V_s^{\lambda,m} - P_s)^{-}  dA^{\lambda,n,\epsilon}_s + \int^T_0(V_s^{\lambda,n} - P_s)^{-}  dA^{\lambda,m,\epsilon}_s + \epsilon A^{\lambda,n,\epsilon}_T + \epsilon A^{\lambda,m,\epsilon}_T.
	\end{align*}
	Since the map $x \mapsto \Phi^+_{\lambda,n}(s, x\vee (P_s+\epsilon))$ is decreasing, 
	\begin{align*}
		& \int^T_t e^{\beta s}(V^{\lambda,n}_s - V_s^{\lambda,m})d(K^{\lambda,n,\epsilon,+}_s - K^{\lambda,m,\epsilon,+}_s)  \\
		& = \int^T_t e^{\beta s}(V^{\lambda,n}_s - V_s^{\lambda,m})(\Phi^+ _{\lambda,n}(s, V^{\lambda, n}_s\vee (P_s+\epsilon))- \Phi^+ _{\lambda,m}(s, V^{\lambda, m}_s\vee (P_s+\epsilon))) ds\\
		& \leq \int^T_t e^{\beta s}(V^{\lambda,n}_s - V_s^{\lambda,m}) (\Phi^+_{\lambda,n}(s, V^{\lambda, m}_s\vee (P_s+\epsilon)) - \Phi^+_{\lambda,m}(s, V^{\lambda, m}_s\vee (P_s+\epsilon)) ) ds
	\end{align*}
	Furthermore, from Young's inequality, we have the following upper bound,
	\begin{align*}
		& \leq \int^T_0 e^{\beta s}\frac{1}{2}(V^{\lambda,n}_s - V_s^{\lambda,m})^2ds\\
		& \quad + \int^T_0 e^{\beta s}\frac{1}{2}(\Phi^+_{\lambda,n}(s, V^{\lambda, m}_s\vee (P_s+\epsilon)) - \Phi^+_{\lambda,m}(s, V^{\lambda, m}_s\vee (P_s+\epsilon)))^2 ds.
	\end{align*}
	
	For the second term in the above, we have
	\begin{align*}
		&\int^T_0 e^{\beta s} (\Phi^+_{\lambda,n}(s, V^{\lambda, m}_s\vee (P_s+\epsilon)) - \Phi^+_{\lambda,m}(s, V^{\lambda, m}_s\vee (P_s+\epsilon)))^2 ds \\
		&\leq e^{\beta T}\int^T_0 \left( \left(\lambda \Phi\left(-\frac{\epsilon}{\lambda / n}\right) + \lambda \ln n\right) - \left(\lambda \Phi\left(-\frac{\epsilon}{\lambda / m}\right) + \lambda \ln m\right)\right)^2 ds \\
		&= e^{\beta T}T \left( \lambda \ln\left(\frac{e^{-\frac{\epsilon}{\lambda / n}}-1}{- \frac{\epsilon}{\lambda / n}} \frac{- \frac{\epsilon}{\lambda / m}}{e^{-\frac{\epsilon}{\lambda / m}}-1} \right) + \lambda \ln \left(\frac{n}{m}\right)\right)^2  \\
		&= e^{\beta T}T\lambda^{2} \left( \ln\left(\frac{e^{-\frac{\epsilon}{\lambda / n}}-1}{e^{-\frac{\epsilon}{\lambda / m}}-1} \right)\right)^2,
	\end{align*}
	
	\noindent where we use the fact that the map $x \mapsto \Phi_{\lambda, n}(s, x\vee (P_{s} + \epsilon))^{+} - \Phi_{\lambda, m}(s, x\vee (P_{s} + \epsilon))^{+}$ is decreasing. In particular, since $n>m$, for any $x \leq P_{s} - \lambda \Phi_{m}^{-1}(0)$, 
	\begin{align*}
		&\Phi_{\lambda, n}(s, x\vee (P_{s} + \epsilon)) - \Phi_{\lambda, m}(s, x\vee (P_{s} + \epsilon)) \\
		&= \lambda \Phi\left(\frac{P_{s} - x \vee(P_{s}+\epsilon)}{\lambda / n}\right) - \lambda \Phi\left(\frac{P_{s} - x \vee(P_{s}+\epsilon)}{\lambda / m}\right) + \lambda \ln\left(\frac{n}{m}\right)  \\
		&= {\lambda}\ln\left(\frac{e^{\frac{P_{s} - x \vee(P_{s}+\epsilon)}{\lambda / n}}-1}{e^{\frac{P_{s} - x \vee(P_{s}+\epsilon)}{\lambda / m}}-1}\right),
	\end{align*}
	which is a decreasing function of the variable $x$ by Lemma \ref{lem: cauchy_increasing}. 
	For $x \geq P_{s} - \lambda \Phi_{m}^{-1}(0)$, the difference reduces to $\Phi_{\lambda, n}(s, x\vee (P_{s} + \epsilon))$ which is also decreasing with respect to $x$ by Lemma \ref{lipphi}.
	
	To estimate the integral against $K^{\lambda,n,\epsilon,-} - K^{\lambda,m,\epsilon,-}$ in equation \eqref{ineq: ito_s2_cauchy_1}, we apply Young's inequality to obtain
	\begin{align*}
		&\int^T_t e^{\beta s}(V^{\lambda,n}_s - V_s^{\lambda,m})d(K^{\lambda,n,\epsilon,-}_s - K^{\lambda,m,\epsilon,-}_s) \\
		&\leq \int^T_t e^{\beta s} (V^{\lambda,n}_s - V^{\lambda,m}_{s})^{2}ds + \int_{t}^{T}e^{\beta s} (\Phi^-_{\lambda,n}(s, V^{\lambda, n}_s) -\Phi^-_{\lambda,m}(s, V^{\lambda, m}_s))^{2} ds.
	\end{align*}
	
	Regarding the local martingale term in \eqref{ineq: ito_s2_cauchy_1}, observe that the process 
	$$
	\int^t_0e^{-\beta(t-s)} (V^{\lambda, n}_{s} - V^{\lambda,m}_{s}) d \bar M^{\lambda,n,m}_s,
	$$ 
	\noindent is a uniformly integrable martingale. This follows from the BDG and Young inequalities, since
	\begin{align*}
		&\EE\left[\left(\int_{0}^{T} e^{-2 \beta(T-s)} (V^{\lambda, n}_{s} - V^{\lambda,m}_{s})^{2} d [\bar M^{\lambda,n,m}]_s\right)^{\frac{1}{2}}\right] \\
		&\leq C \EE\left[\sup_{0\leq s \leq T}(V^{\lambda, n}_{s} - V^{\lambda,m}_{s})^{2} + [\bar M^{\lambda,n,m}]_T\right] \\
		&\leq C \EE\left[\sup_{0\leq s \leq T} (V^{\lambda, n}_{s})^{2} + \sup_{0\leq s\leq T}(V^{\lambda,m}_{s})^{2} + [\bar M^{\lambda,n,m}]_T\right] < \infty.
	\end{align*}
	Taking expectations on both sides of \eqref{ineq: ito_s2_cauchy_1} and combining the above estimates, and choosing $\beta=2$, we obtain the upper bound
	\begin{align} \label{ineq: ito_s2_cauchy_2}
		&\EE \left[\sup_{0 \leq t \leq T}(V_{t}^{\lambda,n} - V^{\lambda,m}_{t})^2 \right] 
		\leq  e^{2T} \EE\left[\int^T_0(V_s^{\lambda,m} - P_s)^{-} dA^{\lambda,n,\epsilon}_s\right]\\
		& \quad + e^{2T} \left[\int^T_0 (V^{\lambda,n}_s - P_s)^{-} dA_s^{\lambda,m,\epsilon}\right]
		+ \frac{e^{2 T}T\lambda^{2}}{2} \left[ \ln\left(\frac{e^{-\frac{\epsilon}{\lambda / n}}-1}{e^{-\frac{\epsilon}{\lambda / m}}-1} \right)\right]^2  \nonumber \\
		& \quad + e^{2T}\EE\left[\int^T_0 (\Phi^-_{\lambda,n}(s, V^{\lambda, n}_s) - \Phi^-_{\lambda,m}(s, V^{\lambda, m}_s))^2 ds\right] \nonumber `\\  
		& \quad + e^{2T}\epsilon \EE[A^{\lambda,n,\epsilon}_T] + e^{2T}\epsilon \EE[A^{\lambda,m,\epsilon}_T].\nonumber 
	\end{align}
	\vskip5pt
	
	\emph{Step 2.} We consider the limit as $n$, $m \rightarrow \infty$ of each term in \eqref{ineq: ito_s2_cauchy_2}. For the first term, recall that 
	\begin{align*}
		V^{\lambda,m}_t = P_T - \int^T_t dM^{\lambda,m}_s + \int^T_t dK^{\lambda,m, \epsilon, +}_s - \int^T_t dK^{\lambda,m, \epsilon, -}_s + \int^T_t dA^{\lambda,m,\epsilon}_s.
	\end{align*}
	We consider the unique solution $(Y^{m, \delta},M^{m,\delta})$ to the following BSDE with Lipschitz driver
	\begin{align*}
		Y^{m, \delta}_{t} = P_{T} - \int_{t}^{T}dM_{s}^{m,\delta} - \int_{t}^{T}\Phi^{-}_{\lambda, m}(s, V^{\lambda, m}_{s})ds + \int_{t}^{T} \beta_{m, \delta}(P_s + \delta - Y^{m, \delta}_s) ds,
	\end{align*}     
	where $\beta_{m, \delta} := m\Phi'\left(-\frac{\delta}{\lambda /m}\right)$ and $\delta < \epsilon$. By the inequality $a\vee b = (b-a)^+ + a$ and the mean-value theorem 
	\begin{align*}
		\Phi_{\lambda,m}(s, Y^{m, \delta}_s)\!- \Phi_{\lambda,m}(s, Y^{m, \delta}_s\vee (P_s+ \epsilon))  
		&\geq \Phi_{\lambda,m}(s, Y^{m, \delta}_s)- \!\Phi_{\lambda,m}(s, Y^{m, \delta}_s\vee (P_s+ \delta)) \\
		& \geq m \Phi'\left(-\frac{\delta}{\lambda /m}\right)(P_s + \delta - Y^{m, \delta}_s).
	\end{align*}  
	Moreover, since $-\Phi^{-}_{\lambda, m}(s, V^{\lambda, m}_{s}) \leq \Phi_{\lambda, m}(s, V^{\lambda, m}_{s} \vee (P_{s} + \epsilon))$, by the comparison theorem we have
	$$
	V^{\lambda, m}_t \geq Y^{m, \delta}_t, \quad t \in [0,T].
	$$ 
	
	From It\^o's formula, we have for any $\tau \in \cT_{0,T}$
	\begin{align*}
		Y^{m, \delta}_{\tau} = \EE\Bigg[ e^{-\beta_{m, \delta}(T-\tau)}P_{T} - \int_{\tau}^{T}e^{-\beta_{m, \delta}(s-\tau)}\left(\Phi^{-}_{\lambda, m}(s, V^{\lambda, m}_{s} ) + \beta_{m, \delta}(P_{s} + \delta) \right) ds 
		\Bigg | \cF_{\tau}\Bigg].
	\end{align*}
	
	\noindent Moreover, since $\beta_{m, \delta} \rightarrow \frac{m}{2}$ as $\delta \rightarrow 0$ and 
	\begin{align*}
		&e^{-\beta_{m, \delta}(T-\tau)}P_{T} - \int_{\tau}^{T}e^{-\beta_{m, \delta}(s-\tau)}\Phi^{-}_{\lambda, m}(s, V^{\lambda, m}_{s} )ds + \beta_{m, \delta}\int_{\tau}^{T}e^{-\beta_{m, \delta}(s-\tau)}(P_{s} + \delta)ds \\
		& \leq P_{T} + m\int_{0}^{T}(P_{s} + \delta)ds
		\leq C\Big(\sup_{0 \leq s \leq T}|P_{s}| + \delta\Big),
	\end{align*}
	
	\noindent the dominated convergence theorem guarantees that
	\begin{align*}
		Y^{m}_{\tau} & := \lim_{\delta \rightarrow 0}Y^{m, \delta}_{\tau}\\
		& \hspace{0.3em} = \EE\Big[e^{-\frac{m}{2}(T-\tau)}P_{T} - \int_{\tau}^{T}e^{-\frac{m}{2}(s-\tau)}\Phi^{-}_{\lambda, m}(s,V^{\lambda, m}_{s} )ds + \frac{m}{2}\int_{\tau}^{T}e^{-\frac{m}{2}(s-\tau)}P_{s} ds \Big | \cF_{\tau}\Big].
	\end{align*}
	
	We now pass to the limit as $m \rightarrow \infty$ in the above limit. It follows from the Cauchy-Schwarz inequality, the Lipschitz regularity of $\Phi_{\lambda, m}^{-}$ and the fact that $V^{\lambda, m} \in \cS^{2}$ that
	\begin{equation*}
		\EE \Big[\int_{\tau}^{T}e^{-\frac{m}{2}(s-\tau)}\Phi^{-}_{\lambda, m}(s, V^{\lambda, m}_{s})ds \Big | \cF_{\tau} \Big] \overset{a.s. \& L^{2}(\mathbb{P})}{\longrightarrow} 0, \quad \mbox{ as } m\rightarrow \infty,
	\end{equation*}
	
	\noindent so that, 
	\begin{align*}
		&Y^{m}_{\tau}=\EE \Big[e^{-\frac{m}{2}(T-\tau)}P_{T} - \int_{\tau}^{T}e^{-\frac{m}{2}(s-\tau)}\Phi^{-}_{\lambda, m}(s, V^{\lambda, m}_{s})ds + \frac{m}{2}\int_{\tau}^{T}e^{-\frac{m}{2}(s-\tau)}P_{s}ds \Big | \cF_{\tau} \Big] \\
		&\overset{a.s. \& L^{2}(\mathbb{P})}{\longrightarrow}  P_{T}\I_{\{\tau = T\}} + P_{\tau} \I_{\{\tau < T\}}\geq P_{\tau}.
	\end{align*}
	
	Hence, by the section theorem (see Theorem 3.2 in \cite{N2006})
	$
	V^{\lambda}  \geq P
	$
	which in turn, using the fact that $V^{\lambda, m} \uparrow V^{\lambda}$, yields
	$
	(V^{\lambda, m}_{s} - P_{s})^{-} \downarrow 0,
	$  a.s.
	
	It follows from Dini's theorem that $\sup_{0\leq s \leq T}(V_{s}^{\lambda, m} - P_{s})^{-} \downarrow 0$. Finally, by the Cauchy-Schwarz inequality, Lemma \ref{lem: Ale_Klne_s2_bound} and the monotone convergence theorem 
	\begin{align*}
		\EE\Big[\int_{t}^{T}(V_{s}^{\lambda, m} - P_{s})^{-}dA_{s}^{\lambda, n, \epsilon}\Big] &\leq \EE\left[\sup_{0\leq s \leq T}(V_{s}^{\lambda, m} - P_{s})^{-} A_{T}^{\lambda, n, \epsilon}\right] \\
		&\leq \Big(\EE\Big[\Big(\sup_{0\leq s \leq T}(V_{s}^{\lambda, m} - P_{s})^{-}\Big)^{2}\Big]\Big)^{\frac{1}{2}} \Vert A^{\lambda,n, \epsilon}\Vert_{\cS^{2}}^{\frac{1}{2}} \rightarrow 0,
	\end{align*}
	as $n, m \rightarrow \infty$. Similarly, we deduce that  $\EE[\int_{t}^{T}(V_{s}^{\lambda, n} - P_{s})^{-}dA_{s}^{\lambda, m, \epsilon}] \rightarrow 0$ as $n$, $m \rightarrow \infty$. 
	
	The third term in \eqref{ineq: ito_s2_cauchy_2} vanishes since for any $\epsilon \geq 0$
	\begin{equation*}
		\lim_{n,m\rightarrow\infty}     \ln\left(\frac{e^{-\frac{\epsilon}{\lambda / n}}-1}{e^{-\frac{\epsilon}{\lambda / m}}-1} \right) = 0.
	\end{equation*}
	
	For the fourth term in \eqref{ineq: ito_s2_cauchy_2}, by \eqref{support:phin:negative:part} and Corollary \ref{cor3.1}, we note that $\Phi^-_{\lambda,n}(s,x)$ is Lipschitz continuous with coefficient independent of $n$ and decreasing in $\lambda$, 
	\begin{align*} 
		(\Phi^-_{\lambda,n}(s, V^{\lambda, n}_s) - \Phi^-_{\lambda,m}(s, V^{\lambda, m}_s))^{2} \leq 4 \max\{1, 1/|\Phi_{2}^{-1}(0)|^{2}\}  (V^{\lambda}_s)^{2}.
	\end{align*}
	Also, in view of \eqref{klambdan} and the Lipschitz continuity of $\Phi_{\lambda, n}^{-}$ and $\Phi_{\lambda, m}^{-}$ with Lipschitz coefficient $C:= \max\{1,1/|\Phi^{-1}_2(0)|\}$ by Corollary \ref{cor3.1}, we have
	\begin{align*} 
		&|\Phi^-_{\lambda,n}(s, V^{\lambda, n}_s) - \Phi^-_{\lambda,m}(s, V^{\lambda, m}_s)|  \\
		&\leq |\Phi^-_{\lambda,n}(s, V^{\lambda, n}_s) - \Phi^{-}_{\lambda, \infty}(s, V_{s}^{\lambda})| + |\Phi^{-}_{\lambda, \infty}(s, V_{s}^{\lambda}) - \Phi^-_{\lambda,m}(s, V^{\lambda, m}_s)| \\
		&\leq \big|\Phi_{\lambda,n}^{-}(s, V^{\lambda,n}_s) - \Phi_{\lambda,n}^{-}(s, V^{\lambda}_s)| + |\Phi_{\lambda,n}^{-} (s, V^{\lambda}_s) -\Phi_{\lambda, \infty}^{-}(s, V^{\lambda}_s)\big| \\
		&\hspace{1em}+ \big|\Phi_{\lambda,m}^{-}(s, V^{\lambda,m}_s) - \Phi_{\lambda,m}^{-}(s, V^{\lambda}_s)| + |\Phi_{\lambda,m}^{-} (s, V^{\lambda}_s) -\Phi_{\lambda, \infty}^{-}(s, V^{\lambda}_s)\big| \\
		&\leq C|V^{\lambda,n}_s - V^{\lambda}_s| +  \big|\Phi_{\lambda,n}^{-}(s, V^{\lambda}_s) -\Phi_{\lambda, \infty}^{-}(s, V^{\lambda}_s)\big|\\
		& \hspace{1em} + C|V^{\lambda,m}_s - V^{\lambda}_s| +  \big|\Phi_{\lambda,m}^{-}(s, V^{\lambda}_s) -\Phi_{\lambda, \infty}^{-}(s, V^{\lambda}_s)\big|,
	\end{align*}
	which converges to zero as $n,m \rightarrow \infty$. Hence, by the dominated convergence theorem 
	\begin{equation*}
		\EE\left[\int^T_0 (-\Phi^-_{\lambda,n}(s, V^{\lambda, n}_s)+ \Phi^-_{\lambda,m}(s, V^{\lambda, m}_s))^2 ds\right] \rightarrow 0, \quad \mbox{ as } m,n \rightarrow \infty.
	\end{equation*}
	
	Finally, by Lemma \ref{lem: Ale_Klne_s2_bound}, we obtain
	\begin{equation*}
		e^{2T}\epsilon \EE[A^{\lambda,n,\epsilon}_T] + e^{2T}\epsilon \EE[A^{\lambda,m,\epsilon}_T]\leq 2Ce^{2T}\epsilon.
	\end{equation*}
	Note that the left-hand side of \eqref{ineq: ito_s2_cauchy_2} does not depend on $\epsilon$. Hence, after sending $n,m \to \infty$, we may let $\epsilon \downarrow 0$ and deduce that $(V^{\lambda,n})_{n\ge 1}$ is Cauchy sequence in $\cS^{2}$.    
\end{proof}

By Lemma \ref{lem: VlgP}, the sequence $V^{\lambda,n}$ converges to $V^{\lambda}$ in $\cS^{2}$ as $n\to\infty$, and therefore the limit inherits the continuity of the sample paths. We thus obtain the following. 

\begin{corollary}
	The process $V^\lambda$ has continuous sample paths.
\end{corollary}

For fixed $n \geq 2$ and $0\leq \lambda \leq 1$, using \eqref{Vln} we may rewrite $V^{\lambda,n}$ as
\begin{equation} \label{eq: V_lambda_n_doob}
	V^{\lambda,n}
	= V^{\lambda,n}_{0} + M^{\lambda,n} - \bigl(K^{\lambda,n,\epsilon} + A^{\lambda,n,\epsilon}\bigr).
\end{equation}
We then observe that $V^{\lambda,n} + K^{\lambda,n,\epsilon}$ is a continuous supermartingale in $\cS^{2}$. By the dominated convergence theorem, we pass the limit as $n \to \infty$ and obtain, for all $s \leq t$,
\begin{align} \label{super}
	\EE\big[V^{\lambda}_t + K^{\lambda,\epsilon}_t \big| \cF_s \big]
	\leq V^{\lambda}_s + K^{\lambda,\epsilon}_s,
\end{align}
where integrability follows from Corollary \ref{cor: VlS2} and Lemma \ref{lem: KleS2}, since $V^{\lambda}$, $K^{\lambda,\epsilon}$, and $P$ all lie in $\cS^{2}$. Consequently, the process $V^{\lambda} + K^{\lambda,\epsilon}$ is a continuous supermartingale of class~(D). 
By Lemma \ref{lem: VlgP}, together with the monotonicity of $K^{\lambda,\epsilon}$ in $\epsilon$, taking the limit as $\epsilon \to 0$ yields
\begin{align} \label{supermartingale}
	\EE\left[V^\lambda_t + K^\lambda_t \middle| \cF_s\right] \leq V^\lambda_s + K^\lambda_s .
\end{align}

Moreover, by Corollary \ref{cor: VlS2} and Lemma \ref{lem: KleS2}, both $V^\lambda$ and $K^\lambda$ lie in $\cS^{2}$; hence $V^\lambda + K^\lambda$ is a process of class (D). The Doob-Meyer decomposition applied to this continuous supermartingale then guarantees the existence of a uniformly integrable martingale $M^\lambda$ and an integrable, non-decreasing process $A^\lambda$ such that
\begin{align*}
	V^\lambda_t = P_T - \int^T_t dM^\lambda_s + \int^T_t \Phi_{\lambda,\infty}(s, V^{\lambda}_s) ds + (A^\lambda_T - A^\lambda_t).
\end{align*}
It remains to prove that $M^\lambda \in \cH^{2}$, and that the non-decreasing process $A^\lambda$ belongs to $\cK^{2}$ and satisfies the Skorokhod reflection condition $\int^T_0 ( V^\lambda_{s} - P_s) dA^\lambda_s = 0$. To this end, we show that the continuous supermartingale $V^{\lambda} + K^{\lambda}$ has an optimal stopping representation.

\begin{theorem}\label{t4.2}
	There exist processes $(M^\lambda, A^\lambda) \in \cH^2 \times \cK^2$ such that 
	\begin{align*}
		V^{\lambda}_t + K^{\lambda}_t & = \esssup_{\sigma \in \mathcal{T}_{t,T}}\mathbb{E}[P_{\sigma} + K^{\lambda}_{\sigma} |\cF_t] = M^\lambda_t - A^\lambda_t,
	\end{align*}
	
	\noindent and the process $A^\lambda$ satisfies the Skorokhod reflection condition $\int^T_0 (V^\lambda_s - P_s) \,  dA^\lambda_s = 0$.
\end{theorem}

\begin{proof}
	
	\noindent \emph{Step 1.} From \eqref{supermartingale} and Lemma \ref{lem: VlgP} we know that $V^\lambda + K^{\lambda}$ is a continuous supermartingale and that $V^\lambda_t \geq P_t$ a.s. for all $t\in [0,T]$. Consequently, $V^\lambda_t + K^{\lambda}_t \geq P_t + K^{\lambda}_t$ and, by the defining property of the Snell envelope associated with the process $P + K^\lambda$, it follows that
	\begin{gather*}
		V^\lambda_t + K^{\lambda}_t \geq \esssup_{\sigma \in \mathcal{T}_{t,T}}\mathbb{E}[P_\sigma + K^{\lambda}_\sigma | \cF_t], \quad a.s.
	\end{gather*}
	
	\vskip5pt
	\emph{Step 2.} To establish the reverse inequality, we introduce the sequence of hitting times $\sigma_n^\epsilon := \inf\{s\geq t: V^{\lambda,n}_s \leq P_s + \epsilon\}\wedge T$. Since the sequence $(V^{\lambda, n})_{n\geq1}$ is non-decreasing, we have $\sigma_n^\epsilon \leq \sigma_{n+1}^\epsilon$ for all $n$. We may therefore define 
	$
	\sigma^\epsilon := \lim_{n\rightarrow \infty}\sigma_n^\epsilon,
	$ 
	\noindent which is a stopping time, as $\{\sigma^\epsilon \leq t\} = \cup_n \{\sigma^\epsilon_n \leq t\} \in \cF_t$. By \eqref{ALE}, we have $A_{\sigma^{\epsilon}_{n}}^{\lambda,n, \epsilon} - A_{t}^{\lambda,n, \epsilon} = 0$. Hence, invoking the decomposition \eqref{eq: V_lambda_n_doob} together with the optional sampling theorem, we obtain 
	\begin{align*}
		V^{\lambda,n}_t + K^{\lambda,n,\epsilon}_t
		& = \EE[M^{\lambda, n}_{\sigma^{\epsilon}_{n}} - A^{\lambda, n, \epsilon}_{t} | \cF_{t}] \\
		& =  \EE[V^{\lambda,n}_{\sigma_{n}^{\epsilon}} + K^{\lambda,n,\epsilon}_{\sigma_{n}^{\epsilon}} |\cF_{t}] \leq \mathbb{E}[P_{\sigma_n^\epsilon} + K^{\lambda,n,\epsilon}_{\sigma_n^\epsilon} |\cF_t] + \epsilon.
	\end{align*}

	It follows from Lemmas \ref{lem:phin_root} and \ref{lip}, together with the inequality $V_s^{\lambda, n}\leq V_s$ that $|\Phi_{\lambda,n} (s, V^{\lambda,n}_s\vee (P_s+ \epsilon))| \leq \frac{\lambda}{\epsilon}\left(|V^\lambda_s| + |P_s| + \lambda\right)$. Combining this estimate with the continuity of the mapping $s\mapsto P_s$, we may pass to the limit as $n\rightarrow \infty$ in the previous inequality. An application of the dominated convergence theorem, together with \eqref{eq3.24}, then yields
	\begin{align*}
		V^{\lambda}_t + K^{\lambda,\epsilon}_t
		& \leq \mathbb{E}[P_{\sigma^\epsilon} + K^{\lambda,\epsilon}_{\sigma^\epsilon} |\cF_t] + \epsilon.
	\end{align*}
	
	Next, we observe that the stopping times $\sigma^\epsilon$ form an increasing family as $\epsilon \in \mathbb{Q}\cap (0,\infty)$ decreases to zero. Consequently the limit $\sigma^*:= \lim_{\epsilon \rightarrow 0} \sigma^\epsilon$ exists and defines a stopping time. In order to pass to the limit as $\epsilon$, we first note that the process $K^{\lambda,\epsilon}$ can be decomposed as the sum of two integrable increasing processes,
	\begin{align*}
		K^{\lambda,\epsilon}_t  
		& = \int^t_0 \Phi_{\lambda,\infty}(s, V^{\lambda}_s\vee (P_s+ \epsilon))\I_{\{V^\lambda_s \leq P_s + \lambda\}} ds + \int^t_0 \Phi_{\lambda,\infty}(s, V^{\lambda}_s)\I_{\{V^\lambda_s > P_s + \lambda\}} ds.
	\end{align*}
	The first integral defines an integrable increasing process whose integrand is monotone in $\epsilon < \lambda$ as $\epsilon \rightarrow 0$, while the second integral is an integrable decreasing process. Therefore, letting $\epsilon \rightarrow 0$ and applying the monotone convergence theorem, we obtain 
	\begin{align*}
		V^{\lambda}_t + K^{\lambda}_t & \leq \mathbb{E}[P_{\sigma^*} + K^{\lambda}_{\sigma^*} |\cF_t] \leq \esssup_{\sigma \in \mathcal{T}_{t,T}}\mathbb{E}[P_{\sigma} + K^\lambda_{\sigma} |\cF_t]. 
	\end{align*}
	Combining the above with the converse bound established earlier, we conclude that 
	$$
	V^{\lambda}_t + K^{\lambda}_t = \esssup_{\sigma \in \mathcal{T}_{t,T}}\mathbb{E}[P_{\sigma} + K^{\lambda}_{\sigma} |\cF_t].
	$$ 	
	By the uniqueness of the Doob–Meyer decomposition and standard results from optimal stopping theory (see, for instance, Theorem D.13 in Karatzas and Shreve \cite{KS1998}), we deduce that the process $A^\lambda$ satisfies the required Skorokhod reflection condition. 
	
	\newpage
	\emph{Step 3.} To establish that $M^\lambda \in \cH^2$ and $A^\lambda\in \cK^2$, we follow the argument of Lemma 3.2 in Grigorova \emph{et al.} \cite{GIOOQ2017}, which relies on Theorem A.2 and Corollary A.1 therein. To this end, we consider the process
	\begin{gather*}
		Y_t := V^\lambda_t + K^\lambda_t  = \esssup_{\sigma \in \mathcal{T}_{t,T}}\mathbb{E}[P_{\sigma} + K^{\lambda}_{\sigma} |\cF_t] = M^\lambda_t - A^\lambda_t.
	\end{gather*}
	
	By Jensen’s inequality, we obtain
	\begin{align*}
		|Y_t| = |\esssup_{\sigma \in \mathcal{T}_{t,T}}\mathbb{E}[P_{\sigma} + K^{\lambda}_{\sigma} |\cF_t] | 
		& \leq \esssup_{\sigma \in \mathcal{T}_{t,T}}\mathbb{E}[|P_{\sigma} + K^{\lambda}_{\sigma}||\cF_t] \\
		& \leq \mathbb{E}[\,\sup_{0\leq t\leq T}|P_{t}| + \sup_{0\leq t\leq T}|K^{\lambda}_{t}||\cF_t].
	\end{align*}
	
	We therefore define $X :=\sup_{0\leq t\leq T}|P_{t}| + \sup_{0\leq t\leq T}|K^{\lambda}_{t}|$. Under assumption \ref{A}, by Lemma \ref{lem: KleS2} the Cauchy-Schwarz inequality, it follows that
	\begin{gather*}
		\EE[X^2] \leq C\|P\|_{\cS^2} + C\|K^\lambda\|_{\cS^2} < \infty.
	\end{gather*}
	Applying Doob’s martingale inequality yields
	\begin{gather*}
		\EE[\,\sup_{0\leq t\leq T} |Y_t|^{2}] \leq \EE[\,\sup_{0\leq t\leq T} |\EE[X|\cF_t]|^2] \leq C\EE[X^2].
	\end{gather*}
	Next, we note that the process $Y_t - \EE[Y_T|\cF_t]$ coincides with the potential generated by $A^\lambda$, namely 
	$$
	Y_t - \EE[Y_T|\cF_t] = \EE[A^\lambda_T - A^\lambda_t|\cF_t].
	$$
	\noindent Moreover,
	$|\EE[Y_T|\cF_t] - Y_t | \leq \EE[|Y_T||\cF_t] + |Y_t|  \leq 2\mathbb{E}[X|\cF_t]$. Invoking Theorem A.2 of \cite{GIOOQ2017}, we therefore deduce the existence of a constant $c>0$ such that $\EE[|A^\lambda_T|^{2}] \leq c\EE[X^2] < \infty$ which shows that $A^{\lambda} \in \cK^{2}$. Finally, since $V^\lambda, K^\lambda \in \cS^{2}$ and $A^\lambda \in \cK^2$, it follows directly that the martingale component $M^{\lambda}$ belongs to $\cH^2$.
\end{proof}

\subsection{Probabilistic Interpretation} \label{subsec: prob_interpretation}
We provide a financial interpretation of the singular reflected BSDE identified in equation \eqref{entropyrbsde}. In particular, we interpret the
value process $V^\lambda$ as the price of an American-style claim subject to default risk.

To this end, we introduce an endogenous default intensity process $\gamma^\lambda$ defined by
\begin{align}
	\gamma^\lambda_s := \frac{\lambda }{P_s + \lambda- V^\lambda_s }  \ln\left(\frac{\lambda }{V^\lambda_s - P_s}\right)\quad \mathrm{and} \quad \Gamma^\lambda_t := \int^t_0 \gamma^\lambda_s ds. \label{intensity_gamma}
\end{align}
The process $\gamma^\lambda$ is strictly positive and therefore defines a valid default intensity. Moreover, it explodes whenever $V^\lambda$ approaches the payoff process $P$, reflecting an imminent default when the continuation value is close $P$. 


For $t\in[0,T]$, we define a default time $\sigma^\lambda_t$ by
$\sigma^\lambda_t  := \inf\{s\geq t: 1-e^{-(\Gamma^\lambda_s - \Gamma^\lambda_t)} \geq U\}$
where $U$ is a uniform random variable on $[0,1]$ independent of $\mathcal{F}_\infty$. By construction, $\sigma^\lambda_t$ is a random time with $\mathbb{F}$-intensity $\gamma^\lambda$. In particular, for any $u \geq t$, it satisfies
\begin{align}
	\mathbb{P}(\sigma^\lambda_t > u \,|\, \mathcal{F}_\infty) = \mathbb{P}(\sigma^\lambda_t > u \,|\, \mathcal{F}_u) = e^{-(\Gamma^\lambda_u - \Gamma^\lambda_t)}, \label{H}
\end{align}

\noindent highlighting the structural link between default risk and early exercise incentives.

\begin{theorem}
	\label{thm:representation}
	The process $V^\lambda$ admits the representation for each $t\in[0,T]$ 
	\begin{equation}
		\label{eq:main-rep}
		V_t^\lambda 
		= \esssup_{\tau_t\in\mathcal{T}_{t,T}}
		\mathbb{E}\Big[ P_{\tau_t} \I_{\{\sigma_t^\lambda>\tau_t\}} 
		+ (P_{\sigma_t^\lambda}+\lambda)\,\I_{\{\sigma_t^\lambda\le \tau_t\}} \;\Big|\;\mathcal{F}_t\Big].
	\end{equation}
	Consequently, $V^\lambda$ is the value of a defaultable American option with exercise payoff $P$ and recovery payoff $P+\lambda$, where default occurs at $\sigma^\lambda_t$ with intensity $\gamma^\lambda$.
\end{theorem}

\begin{proof}
	In view of \eqref{intensity_gamma}, the backward dynamics of $V^\lambda$ can be written as 
	\begin{align*}
		V_t^\lambda = P_T + \int^T_t (P_s + \lambda - V_s^\lambda) \gamma^\lambda_s \, ds  - (M^\lambda_T - M^\lambda_t) + A^\lambda_T - A^\lambda_t.
	\end{align*}
	Define the stopping time $\tau^*_t := \inf\{s\geq t : V^\lambda_s = P_s\}$ and by applying It\^o's formula to the process $e^{-\Gamma^\lambda_s}V^\lambda_s$, and taking conditional expectations yields
	\begin{align*}
		e^{-\Gamma_t^\lambda}V^\lambda_t
		&= \mathbb{E}[P_{\tau^*_t} e^{-\Gamma_{\tau^*_t}^\lambda} + \int^{\tau^*_t}_t  e^{-\Gamma_s^\lambda}(P_s +\lambda) \gamma_s^{\lambda} ds \mid \mathcal{F}_t].
	\end{align*}
	where we used the fact that the reflection process $A^\lambda$ does not increase on $[t,\tau^*_t]$. Using the conditional survival probability \eqref{H}, it follows that
	\begin{align*}
		V^{\lambda}_t & \leq \esssup_{\tau_t \in \mathcal{T}_{t,T}}\mathbb{E}[P_{\tau_t}e^{-(\Gamma_{\tau_t}^\lambda-\Gamma_t^\lambda)}+ \int^{\tau_t}_t (P_s + \lambda) \gamma_s^{\lambda} e^{-(\Gamma_s^\lambda-\Gamma_t^\lambda)}ds|\mathcal{F}_t]\\
		& = \esssup_{\tau_t \in \mathcal{T}_{t,T}}\mathbb{E}[P_{\tau_t}\I_{\{\sigma^\lambda_t > \tau_t\}} + \I_{\{\sigma^\lambda_t\leq \tau_t\}}(P_{\sigma^\lambda_t} + \lambda)|\mathcal{F}_t]
	\end{align*}
	
	To prove the reverse inequality, let $\tau_t \geq t$ be an arbitrary $\mathbb{F}$-stopping time. Repeating the above computation up to $\tau_t$ yields
	\begin{align*}
		e^{-\Gamma_t^\lambda}V^\lambda_t
		&\geq  \mathbb{E}[P_{\tau_t} e^{-\Gamma_{\tau_t}^\lambda} + \int^{\tau_t}_t  e^{-\Gamma_s^\lambda}(P_s +\lambda) \gamma_s^{\lambda} ds \mid \mathcal{F}_t].
	\end{align*}
	where we have used the fact that $A^\lambda$ is increasing and $P \leq V^\lambda$. Combining the above computations, we deduce that
	\begin{align*}
		V^{\lambda}_t
		& = \esssup_{\tau_t \in \mathcal{T}_{t,T}}\mathbb{E}[P_{\tau_t}\I_{\{\sigma^\lambda_t > \tau_t\}} + \I_{\{\sigma^\lambda_t\leq \tau_t\}}(P_{\sigma^\lambda_t} + \lambda)|\mathcal{F}_t].
	\end{align*}
\end{proof}


\section{Numerical Experiments}\label{NEP}
In this section, we illustrate the practical feasibility of our methodology on a simple low-dimensional example with $d=2$.
We consider the symmetric case of an American max-call option.
Specifically, the underlying assets are assumed to follow a $d$-dimensional Black--Scholes model with dividends,
\begin{equation} \label{eq: GBMwD}
	S_t^i = S_0^i \exp\big((r - \delta - \sigma^2/2) t + \sigma W_t^i\big), \quad i = 1, \dots, d,
\end{equation}

\noindent where $S_0^i$ denotes the initial asset prices, $r$ the risk-free interest rate, $\delta$ the constant dividend yield, $\sigma$ the volatility parameter, and
$W=(W^1,\dots,W^d)$ a standard $d$-dimensional Brownian motion.
Given a strike price $K$, the value of the American max-call option is
\begin{align*}
	\sup_{\tau \in \mathcal{T}_{0,T}} \mathbb{E}\!\left[e^{-r\tau} \Big(\max_{1 \leq i \leq d} S_{\tau}^{i} - K\Big)^{+}\right].
\end{align*}

Throughout the numerical experiment, we fix the parameters
\[
S_0^1 = S_0^2= S_0, \quad K = 50, \quad r = 0.05, \quad \sigma = 0.2, \quad \delta = 0.1, \quad T = 3.
\]
We discretise the time interval $[0,T]$ using the uniform time grid $t_{k}=k\Delta t $, $k=0,\cdots, N$, with mesh size $\Delta t = T/N$. We fix $N=100$ in our numerical experiments. We compare the prices obtained from the classical penalisation scheme~\eqref{pscheme}, the entropy-regularised penalised BSDE~\eqref{Vln}, and the PIA defined in~\eqref{eq: pin2}-\eqref{eq: pin3}, against a binomial tree benchmark. We note that the BSDE formulations introduced above do not explicitly include discounting. Applying It\^o's formula to the randomized stopping representation~\eqref{rstopping} shows that discounting by a constant rate $r$ corresponds to adding a $-rV$ term to the driver. 


	\subsection{Numerical Implementation of the PIA} \label{mIPIA}
	We initialize the PIA with $\mathscr{V}^{\lambda,0}=P_0+1$, to ensure that the initial value is strictly positive.
	Given $\mathscr{V}^{\lambda,m}$, the next iterate $\mathscr{V}^{\lambda,m+1}$ satisfies the linear BSDE
	\begin{align*}
		\mathscr{V}^{\lambda,m+1}_t  & = P_T - (\mathscr{N}^{\lambda,m+1}_T - \mathscr{N}^{\lambda,m+1}_t) + \int^T_t  \left\{ G(s, \mathscr{V}^{\lambda,m+1}_{s} , \pi^{m+1}_{s})  -r\mathscr{V}^{\lambda,m+1}_s \right\} \, ds, \nonumber
	\end{align*}
	where $G$ is defined in~\eqref{eq: Gpi_def}.
	By Theorem~3.3 in~\cite{OZ}, this BSDE satisfies
	\begin{align*}
		\mathscr{V}^{\lambda,m+1}_t & =\mathbb{E} \Big[ e^{-\int_t^T a^m_u \, du} \, P_T + \int_t^T e^{-\int_t^s a^m_u  \, du} \,b^m_s ds \Big| \mathcal{F}_t \Big]
	\end{align*}
	
	\noindent so that 
	$$
	\mathscr{V}^{\lambda,m+1}_{t_k}  = \mathbb{E}\Big[ e^{-\int^{t_{k+1}}_{t_k} a^{m}_{s} ds}  \mathscr{V}^{\lambda,m+1}_{t_{k+1}} + \int_{t_k}^{t_{k+1}} e^{-\int_{t_k}^s a^{m}_u \, du} \, b^{m}_s \, ds  \,\Big|\, \mathcal{F}_{t_k} \Big], \quad k=0, \cdots, N-1.
	$$
	We thus consider the approximation scheme
	\begin{align}
		\widehat{\mathscr{V}}_{t_k}^{\lambda, m+1}& = e^{-a_{t_k}^{m}\Delta t} \EE\left[\widehat{\mathscr{V}}^{\lambda,m+1}_{t_{k+1}} \Big | \cF_{t_k}\right] + \frac{b^{m}_{t_k}}{a^{m}_{t_k}}(1 - e^{-a^{m}_{t_k}\Delta t}), \quad k=N-1, \cdots, 0.
		\label{eq: PIA_linear}
	\end{align}
	where
	\begin{align*}
		a^m_{t_k} = \mu_{\pi^{m+1}_{t_k}}+r \qquad \text{and} \qquad b^m_{t_k} = 
		\lambda \Phi\!\Big( \frac{P_{t_k} - \mathscr{V}^{\lambda,m}_{t_k}}{\lambda/n} \Big)
		+ \lambda \ln(n)
		+ \mathscr{V}^{\lambda,m}_{t_k} (a^{m}_{t_k}-r).
	\end{align*}
	The conditional mean $\mu_{\pi^{m+1}_{t_k}}$ admits the explicit expression
	\begin{equation}
		\mu_{\pi^{m+1}_{t_k}} = \mu(\alpha_{t_k}^m,n) 
		:= \frac{n}{1 - e^{-\alpha_{t_k}^m n}} - \frac{1}{\alpha_{t_k}^m} \quad \mathrm{where} \quad \alpha_{t_k}^m = \frac{P_{t_k} - \widehat{\mathscr{V}}^{\lambda,m}_{t_k}}{\lambda}.
	\end{equation}
	Therefore, each policy update reduces to a regression step, which is iterated backwards in time for each $m$. Finally, all methods require the estimation of conditional expectations.
	We employ a least-squares regression based on the 13 basis functions proposed by Andersen and Broadie~\cite{AB2004}.

	We report in Table~\ref{table:results} the numerical prices obtained for the American max-call option using the entropy-regularised implicit BSDE solver (see $\theta$-scheme in for instance, Lionnet \emph{et al.}~\cite{LRS2015}). and the PIA.
	For comparison, we also include the results obtained from the classical penalisation approach of El Karoui \emph{et al.} \cite{EKPPQ1997} and a binomial tree approximation, which serves as a benchmark.
	All entropy-based methods are implemented with temperature parameter $\lambda=1/n$, and prices are reported for several values of the initial asset price $S_0$ and truncation level $n$.

	\begin{table}[ht]
		\centering
		\begin{tabular}{r r |>{\centering\arraybackslash}m{2.5cm} >{\centering\arraybackslash}m{2.5cm}
				|>{\centering\arraybackslash}m{2.0cm} >{\centering\arraybackslash}m{2.2cm}}
			{$S_0$} & {$n$} & {Implicit solver} & {PIA} & {Classical penalization} & {Binomial} \\
			\hline
			\hline
			90 & 10 & 7.388 & 7.463 & 8.208 & 8.296 \\
			90 & 100 & 8.150 & 8.231 & 8.424 & 8.296 \\
			90 & 1000 & 8.285 & 8.367 & 8.460 & 8.296 \\
			\hline
			100 & 10 & 13.246 & 13.349 & 14.040 & 14.211 \\
			100 & 100 & 14.086 & 14.213 & 14.357 & 14.211 \\
			100 & 1000 & 14.227 & 14.350 & 14.408 & 14.211 \\
			\hline
			110 & 10 & 20.821 & 20.926 & 21.494 & 21.799 \\
			110 & 100 & 21.678 & 21.814 & 21.914 & 21.799 \\
			110 & 1000 & 21.815 & 21.926 & 21.980 & 21.799 \\
			\hline
		\end{tabular}
		\caption{Results for American max-call option using implicit solver and policy improvement compared to classical penalization and binomial tree. The temperature parameter is set as $\lambda = 1/n$. The implicit solver uses 20 steps of Newton iterations, and the PIA is computed over 10 iterations. }\label{table:results}
	\end{table}
	
	Several observations follow from Table~\ref{table:results}. First, the
	entropy-regularized method, computed via the implicit solver or the
	PIA—produces prices that converge monotonically in $n$ and consistently
	approach the binomial benchmark. By contrast, the classical penalization
	scheme converges more slowly and tends to overestimate the option value.
	Second, the PIA yields prices that are systematically (slightly) higher than
	those from the implicit solver, in line with its theoretical interpretation
	as an increasing sequence of approximations to the entropy-regularized value
	function. Finally, even for moderate $n$, both entropy-based approaches provide accurate
	approximations to the benchmark, illustrating the practical efficiency and
	numerical stability of the methodology.

	\newpage
	\section{Appendix}\label{AL}
	
	The proofs of the first three lemmas are provided in the appendix of \cite{CFL2025}.

	\begin{lemma} \label{lem:psi:cdf}
		The function $\Psi$ defined in \eqref{def:psi:phi} is a cumulative distribution function.
	\end{lemma}
	
	\begin{lemma}\label{lipphi}
		The function $\Phi$ defined in \eqref{def:psi:phi} satisfies the properties that $ 0 \leq \Phi'(x)\leq 1$. In particular, $\Phi$ is Lipschitz continuous with Lipschitz constant~$1$.
	\end{lemma}


	\begin{lemma}\label{lemma1.1}
		For any $\varepsilon\in(0,1)$ and $c>0$, the following inequality holds:
		\begin{equation}\label{eq1.1}
			0 \le x^{+}-c\,\Phi(x/c)
			\le \varepsilon
			- c\ln\bigl(1-e^{-\varepsilon/c}\bigr)
			+ c[\ln|x|]^{+}
			- c\ln c .
		\end{equation}
	\end{lemma}

	
	The proof of the following result being similar to the one of Lemma \ref{lipphi} is omitted.
	
	\begin{lemma} \label{lem:prop:phi:n}
		The function $\Phi_n$ defined in \eqref{def:Phi:n} satisfies the property that $ 0 \leq \Phi'_n(x)\leq n$ and for any $x$ it is monotonically increasing in $n$.
	\end{lemma}
	

	\begin{lemma} \label{lem: cauchy_increasing}
		Let $n>m$. The function 
		$
		f(x)=\frac{x^{n}-1}{x^{m}-1}
		$
		is increasing on $(0,1)$.
	\end{lemma}
	
	\begin{proof}
		A straightforward computation yields
		\begin{equation*}
			f'(x) = \frac{(n-m)x^{n+m-1} - nx^{n-1}+mx^{m-1}}{(x^{m}-1)^{2}} = \frac{x^{m-1}[(n-m)x^{n} - nx^{n-m}+m]}{(x^{m}-1)^{2}}.
		\end{equation*}
		Since $x^{m-1}>0$ and $(x^{m}-1)^2>0$ for $x\in(0,1)$, the sign of $f'(x)$ is determined by the function $g(x) := (n-m)x^{n} - nx^{n-m}+m$ for $x\in (0,1)$. Differentiating, we obtain $g'(x) = n(n-m)x^{n-1} - n(n-m)x^{n-m-1} = n(n-m)x^{n-m-1}(x^{m} - 1) < 0$ so that $g$ is strictly decreasing on $(0,1)$. Moreover, $\lim_{x \rightarrow 1^{-}}g(x) = 0$, which implies that $g(x)\ge 0$ for all $x\in(0,1)$. Consequently, $f'(x)\ge 0$ on $(0,1)$, and the result follows.
	\end{proof}

\end{document}